%% file: main.tex
\def\fullversionbuild{}
\documentclass[sigconf]{acmart}

\makeatletter
\@ACM@balancefalse
\makeatother
\usepackage{mathtools}
\usepackage[ruled,vlined,linesnumbered]{algorithm2e}
\usepackage{xspace}
\usepackage{booktabs}
\usepackage{multirow}
\usepackage{enumitem}

\usepackage{tikz}
\usetikzlibrary{positioning,arrows.meta,fit,backgrounds}

\acmConference[WWW '27]{Proceedings of the ACM Web Conference 2027}{2027}{ACM}
\acmBooktitle{Proceedings of the ACM Web Conference 2027 (WWW '27)}
\renewcommand\footnotetextcopyrightpermission[1]{}

\AtBeginDocument{%
  \setlength{\abovedisplayskip}{3pt plus 1pt minus 1pt}%
  \setlength{\belowdisplayskip}{3pt plus 1pt minus 1pt}%
  \setlength{\abovedisplayshortskip}{2pt plus 1pt minus 1pt}%
  \setlength{\belowdisplayshortskip}{2pt plus 1pt minus 1pt}%
}

\DontPrintSemicolon
\SetKwInOut{Parameter}{Parameter}

\input{macros}

\AtBeginDocument{%
  \newtheorem{assumption}[theorem]{Assumption}%
  \theoremstyle{remark}%
  \newtheorem{remark}[theorem]{Remark}%
}

\begin{document}
\raggedbottom

\title{Geo-Filtered Approximate Nearest-Neighbor Search on a Single Proximity Graph}

\author{Zhiqiu Zou}
\affiliation{%
  \institution{Beijing Institute of Technology}
  \city{Beijing}
  \country{China}
}
\email{fall@bit.edu.cn}

\author{Rong-Hua Li}
\affiliation{%
  \institution{Beijing Institute of Technology}
  \city{Beijing}
  \country{China}
}
\email{lironghuabit@126.com}

\author{Qiangqiang Dai}
\affiliation{%
  \institution{Beijing Institute of Technology}
  \city{Beijing}
  \country{China}
}
\email{qiangd66@gmail.com}

\author{Guoren Wang}
\affiliation{%
  \institution{Beijing Institute of Technology}
  \city{Beijing}
  \country{China}
}
\email{wanggrbit@126.com}

\renewcommand{\shortauthors}{Zhiqiu Zou, Rong-Hua Li, Qiangqiang Dai, and Guoren Wang}

\input{sections/00_abstract}



\maketitle

\input{sections/01_introduction}
\input{sections/02_preliminaries}
\input{sections/03_recovery_graph}
\input{sections/04_reach_index}
\input{sections/05_experiments}
\input{sections/06_related_work}
\input{sections/07_conclusion}

\clearpage
\setlength{\bibsep}{0pt}
\bibliographystyle{ACM-Reference-Format}
\bibliography{references}

\ifwithappendix
  \clearpage
  \appendix
  \input{appendix/theory_proofs}
  \input{appendix/additional_details}
\fi

\end{document}

%% file: macros.tex
\newif\ifwithappendix
\ifdefined\fullversionbuild
  \withappendixtrue
\else
  \withappendixfalse
\fi

\newcommand{\apx}[1]{\ifwithappendix Appendix~\ref{#1}\else the full version\fi}

\newcommand{\fullversionlink}{%
  \ifwithappendix\else
    \footnote{\url{https://anonymous.4open.science/r/GeoSearch-new-1730/full-version.pdf}}%
  \fi}

\newcommand{\method}{REACH\xspace}
\newcommand{\problem}{GFANN\xspace}
\newcommand{\dataset}{\mathcal{D}}
\newcommand{\Aset}{\mathcal{A}}
\newcommand{\RecG}{\mathcal{H}}
\newcommand{\rowu}{\operatorname{row}}

\newcommand{\rk}{\ell}

\newcommand{\dvec}{d}
\newcommand{\dgeo}{\Delta}

%% file: sections/00_abstract.tex
\begin{abstract}
Geo-filtered approximate nearest-neighbor (\problem) search retrieves the objects inside a query-specified geographic disk whose feature vectors are closest to a query vector, a core operation in location-aware Web search. A disk filter on a standard proximity graph invalidates out-of-disk routing vertices; on selective disks, these removals break greedy paths and recall drops sharply. Existing methods scan all in-disk objects, over-fetch unfiltered candidates, or maintain separate sub-indexes across spatial partitions, trading latency for index overhead. We propose \method, a single-graph index based on \emph{row-local active-order invariance}: if every vertex stores its outgoing neighbors in a fixed-width row sorted by vector distance, any disk filter masks out-of-disk entries while surviving in-disk (\emph{active}) neighbors remain vector-sorted. Offline, each vertex merges globally closest vector candidates with a nearest-first geographic stream into one row; online, beam search evaluates vector distances only for the first $\rk$ active entries of each visited row. Because co-located objects share the same distance to every query center, every disk admits a coordinate group in full or not at all, decomposing active connectivity into intra-group and inter-coordinate conditions. Under a spatial marked-site model with repeated coordinates, we prove that polylogarithmic out-degree guarantees strong connectivity across all disks and constant-factor distance reduction per hop on compact disks, while a cut-rank threshold characterizes the minimum scan budget $\rk$. Across four datasets up to $10^7$ objects, one \method index per dataset outperforms SOTA methods, achieving $2.49\times$ higher geometric-mean QPS than the strongest baseline at $0.95$ Recall@10 on selective queries, while building $3.8\times$--$11.9\times$ faster and using $37\%$--$45\%$ less space than a SOTA spatial-graph index.
\end{abstract}

%% file: sections/01_introduction.tex
\section{Introduction}
\label{sec:introduction}

Approximate nearest-neighbor (ANN) search in high-dimensional vector spaces is a core operation in Web search, recommendation, and retrieval-augmented generation~\cite{Pan_2024,https://doi.org/10.48550/arxiv.2005.11401}. Proximity graphs such as HNSW~\cite{hnsw}, NSG~\cite{nsg}, and DiskANN~\cite{diskann} achieve the strongest trade-off between recall and latency~\cite{DBLP:journals/pvldb/WangXY021} by linking each object to proximate vector neighbors and routing queries through best-first beam search. In location-aware Web services, vector similarity is often paired with a spatial range filter: a user may search for geotagged photos visually similar to a query image within a kilometer of a landmark, or retrieve points of interest (POIs) matching a text query within walking distance. We call this problem \emph{geo-filtered approximate nearest-neighbor} (\problem) search: given a query vector and a query-centered geographic disk, return the $k$ in-disk objects closest to the query in vector space.

\begin{figure*}[t]
  \centering
  \includegraphics[width=0.78\textwidth]{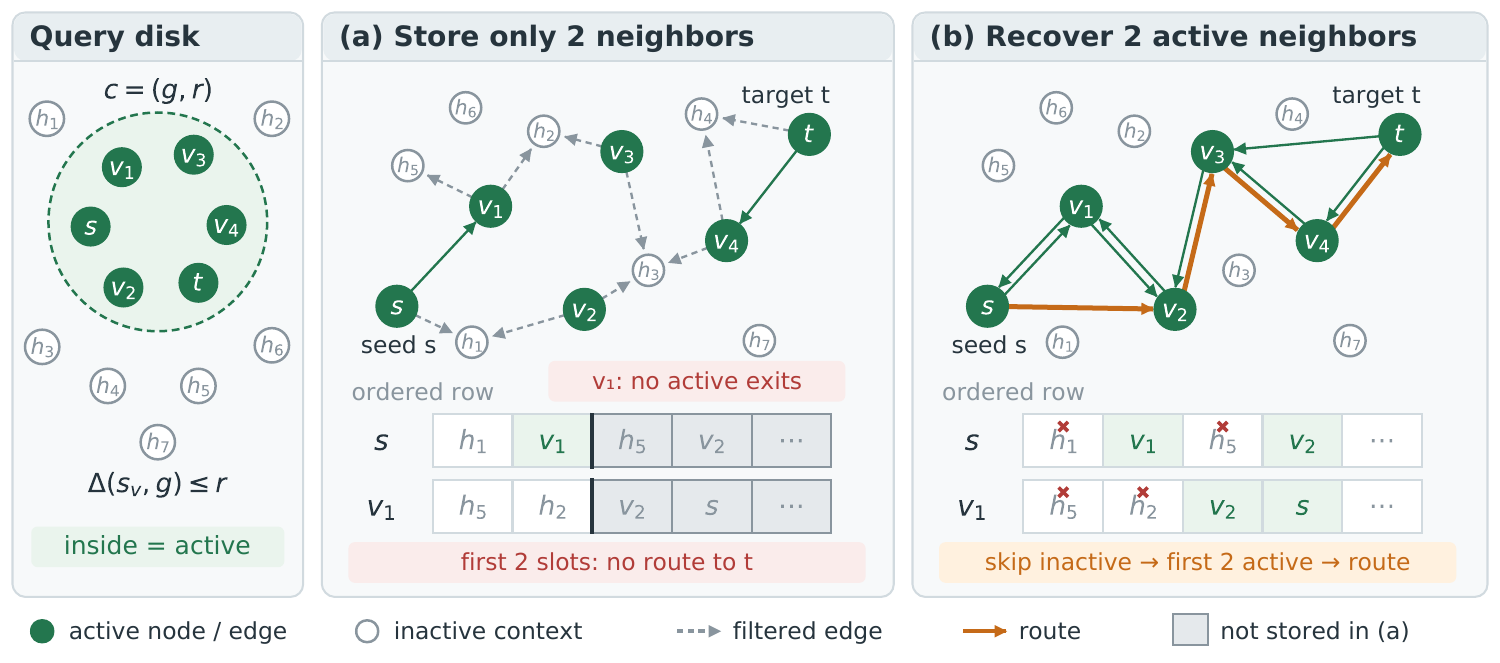}
  \caption{Filter-induced path erosion and active-edge recovery from entry seed $s$ to target $t$ under budget $\rk=2$ (Example~\ref{ex:erosion}). Left: disk $c=(g,r)$ activates $\Aset(c)=\{s,v_1,v_2,v_3,v_4,t\}$ (green) and deactivates $\{h_1,\ldots,h_7\}$ (grey). Cards~(a) and~(b) place these objects in vector space alongside the stored rows of $s$ and $v_1$: \emph{(a)}~storing only two neighbors truncates $\rowu(v_1)$ to $[h_5,h_2]$, leaving $v_1$ with no active exits; \emph{(b)}~scanning wider ordered rows skips inactive entries ($h_1,h_5$ in $\rowu(s)$ and $h_5,h_2$ in $\rowu(v_1)$) to recover the first $\rk=2$ active neighbors ($[v_1,v_2]$ for $s$, $[v_2,s]$ for $v_1$), forming $\RecG_2(c)$ and restoring route $s\to v_2\to v_3\to v_4\to t$ (orange).}
  \Description{A separate geographic disk selects six labeled active vertices. Filtering two-entry rows traps the route at v1; recovering two active entries restores s, v2, v3, v4, t.}
  \label{fig:recovery-graph}
\end{figure*}

Applying a geographic disk filter to a standard proximity graph breaks its navigation structure. Because semantic similarity is independent of physical location, a vertex's closest vector neighbors are scattered globally; restricting traversal to a disk removes out-of-disk routing vertices, which severs multi-hop paths and disconnects the in-disk subgraph (\emph{filter-induced path erosion}, Section~\ref{sec:preliminaries}). Existing filtered ANN methods struggle as disk selectivity varies~\cite{Chronis_2025}. Filter-agnostic strategies fail at one extreme: (1)~\emph{pre-filtering} (\emph{Spatial-First}~\cite{chen2021survey}) scans all in-disk objects, which is fast on tiny radii but slows by orders of magnitude on neighborhood- or city-scale disks ($10^2$--$10^6$ objects); (2)~\emph{post-filtering} (\emph{Vector-First}, VBASE~\cite{vbase}) retrieves global vector candidates before spatial filtering, wasting most distance evaluations on selective disks; and (3)~\emph{in-filtering} (ACORN~\cite{acorn}) explores two-hop unfiltered neighborhoods, which still disconnect under selective disks. Predicate-aware indexes encode filters into the graph for discrete labels~\cite{filteredDiskANN,ung,jag,cgif}, workload-tuned sub-indexes~\cite{sieve}, weighted attribute--vector metrics~\cite{nhq,deg}, or 1D numeric intervals~\cite{serf,window,urng,rnsg,Xu_2024}. However, 1D range graphs rely on a global attribute order under which every interval is contiguous; 2D disks admit no such order even on three objects and induce $\Omega(n^2)$ distinct active sets (Section~\ref{sec:no-per-condition}), whereas spatial partitions~\cite{DBLP:journals/vldb/SongYCYXLC25,https://doi.org/10.48550/arxiv.2604.06616} either split queries across cell graphs or duplicate edges across tree levels.

In this paper, we show that a \emph{single} static graph of fixed out-degree can support arbitrary geographic disks through \emph{Ranked Active-Graph Recovery}. Rather than sorting the dataset along a global axis, we use \emph{row-local active-order invariance}: each vertex $u$ stores its $B$ outgoing neighbors in a fixed-width \emph{row} $\rowu(u)$ pre-sorted by vector distance from $u$. Under any disk $c$, out-of-disk entries in $\rowu(u)$ are ignored while surviving in-disk (\emph{active}) entries remain sorted by vector distance. Retaining all active entries defines the complete active projection $\RecG_{\infty}(c)$, which sets the recall ceiling on $c$; scanning $\rowu(u)$ sequentially and taking only the first $\rk$ active entries reconstructs a degree-$\rk$ \emph{recovery graph} $\RecG_{\rk}(c)\subseteq\RecG_{\infty}(c)$ that caps per-step vector evaluations at $\rk$, with minimum connectivity budget given by the \emph{cut-rank threshold} $\rk_{\mathrm{SC}}(c)$.

To keep $\RecG_{\infty}(c)$ connected across all $\Omega(n^2)$ disks on both POI catalogs with mostly unique coordinates and photo archives where thousands of records share a landmark coordinate~\cite{Antoniou_2010}, we group co-located objects into \emph{coordinate fibers}. If a vertex simply linked to its geographically closest objects, a dense landmark would consume its entire neighbor budget on co-located ties. Because all members of a fiber lie at the same distance from any query center, any disk admits a fiber \emph{either in full or not at all}. This property decomposes active connectivity into query-independent \emph{intra-fiber connectivity} and \emph{inter-coordinate quotient connectivity}, both satisfied at polylogarithmic degree under a marked-site model. Guided by this analysis, \method builds one width-$B$ row per vertex offline from global vector candidates and a fiber-aware geographic stream, and answers queries online through an $O(n)$ spatial seed grid and a two-phase \emph{filter-then-score} kernel. Our contributions are:
\begin{itemize}[leftmargin=1.4em,nosep]
  \item \textbf{Theoretical Foundation:} We prove that 2D disks induce $\Omega(n^2)$ active sets (Theorem~\ref{thm:range-complexity}) and admit no contiguous order (Proposition~\ref{prop:no-order}). We then define $\RecG_{\rk}(c)$ and cut rank $\rk_{\mathrm{SC}}(c)$ (Theorem~\ref{thm:cut-budget}), prove fiber closure and route decomposition (Lemma~\ref{lem:fiber-closure}, Theorem~\ref{thm:factorization}), and establish polylogarithmic degree bounds for strong connectivity and compact-disk $\alpha$-shortcut routing (Theorem~\ref{thm:end-to-end}, Corollary~\ref{cor:navigability}).
  \item \textbf{Graph Index and Query Algorithms:} We develop \method's dual-source row assembly (Algorithm~\ref{alg:row-assembly}), which combines global vector candidates with a 3D kd-tree geographic stream, bounds coordinate ties through intra-fiber $\alpha$-pruning, and reuses vector survivors as zero-cost witnesses (Proposition~\ref{prop:assembly-contract}), paired with a two-phase first-$\rk$-active expansion algorithm (Algorithm~\ref{alg:expand}) and an $O(n)$ grid seed directory.
  \item \textbf{Comprehensive Evaluation:} Across four datasets up to $10^7$ objects and six selectivities ($\bar{\sigma}\in[10^{-5},\allowbreak 1]$), one \method index per dataset achieves $2.49\times$ higher geometric-mean QPS than the best per-setting baseline at $0.95$ Recall@10 on selective queries (fastest in $14/16$ selective cells, $6.33\times$ over CubeGraph, $3.56\times$ over Geo-Partition), builds $3.8\times$--$11.9\times$ faster than CubeGraph, and uses $37\%$--$45\%$ less space on three of the four datasets.
\end{itemize}

%% file: sections/02_preliminaries.tex
\section{Preliminaries}
\label{sec:preliminaries}

\noindent\textbf{Data Model and \problem Queries.}
Let $\dataset=\{o_v=(x_v,s_v):v\in V\}$ be a dataset of $n=|V|$ objects with identifiers $V=\{1,\ldots,n\}$. Each object $v\in V$ has a feature vector $x_v\in\mathcal{X}\subseteq\mathbb{R}^{D}$ measured by Euclidean distance $\dvec(x,y)=\lVert x-y\rVert_2$ and a geographic coordinate $s_v=(\lambda_v,\phi_v)\in\mathcal{S}\subseteq\mathbb{S}^2$ measured by great-circle distance $\dgeo(s,t)$ (notation in \apx{app:notation}\fullversionlink). A spatial filter $c=(g,r)$ with center $g\in\mathcal{S}$ and radius $r>0$ defines the closed disk $\Omega_c=\{s\in\mathcal{S}:\dgeo(s,g)\le r\}$ and the \emph{active set} $\Aset(c)=\{v\in V:\dgeo(s_v,g)\le r\}$ of size $m=|\Aset(c)|$ and selectivity $\sigma(c)=m/n$. Given a \problem query $q=(x_q,c_q,k)$ with $k_q=\min\{k,m\}$, the exact answer $T_k(q)\subseteq\Aset(c_q)$ consists of the $k_q$ active objects closest to $x_q$ under $v \prec_q w \iff (\dvec(x_q,x_v),v) <_{\mathrm{lex}} (\dvec(x_q,x_w),w)$.

\noindent\textbf{Proximity Graphs, Beam Search, and Path Erosion.}
\label{sec:path-erosion}
To avoid scanning all $m$ active vectors in $O(mD)$ time, ANN systems build a directed proximity graph $G=(V,E)$ offline~\cite{hnsw,nsg,diskann}: each vertex $u$ collects vector-proximate candidates and applies $\alpha$-witness diversification (discarding $w$ if an earlier neighbor $v$ satisfies $\alpha\,\dvec(x_v,x_w)\le\dvec(x_u,x_w)$) to cap out-degree at $B$~\cite{DBLP:journals/pvldb/WangXY021}, forming $u$'s \emph{row} $\rowu(u)$. At query time, best-first \emph{beam search} maintains a pool of at most $L$ candidates ordered by $\dvec(x_q,\cdot)$ and expands the closest unvisited vertex until no pool candidate is closer than the $k$-th best result ($\mathrm{Recall}@k=|R\cap T_k(q)|/k_q$). Effective search requires \emph{directed reachability} (which sets the recall ceiling) and \emph{distance-contracting edges} (which govern search speed). Under query $q=(x_q,c_q,k)$, skipping out-of-disk vertices restricts traversal to the \emph{active projection} $G[\Aset(c_q)] = (\Aset(c_q), \{(u,v)\in E : u,v\in\Aset(c_q)\})$. When vectors and locations are independent, each global edge lands in $c_q$ with probability $\sigma(c_q)$, leaving only $B\sigma(c_q)$ active edges per vertex in expectation ($0.032$ edges at $\sigma=10^{-3}, B=32$) and disconnecting $G[\Aset(c_q)]$.

\begin{example}
\label{ex:erosion}
In Figure~\ref{fig:recovery-graph}, disk $c=(g,r)$ admits $\Aset(c)=\{s,v_1,v_2,v_3,v_4,t\}$ (green) and excludes $\{h_1,\ldots,h_7\}$ (grey). In card~(a), storing two vector neighbors leaves $\rowu(s)=[h_1,v_1]$ with only $s\to v_1$ and $\rowu(v_1)=[h_5,h_2]$ with no active exits, which traps traversal at $v_1$. In card~(b), wider sorted rows ($\rowu(s)=[h_1,\allowbreak v_1,\allowbreak h_5,\allowbreak v_2,\ldots]$ and $\rowu(v_1)=[h_5,\allowbreak h_2,\allowbreak v_2,\allowbreak s,\ldots]$) recover the first $\rk=2$ active entries ($[v_1,v_2]$ at $s$ and $[v_2,s]$ at $v_1$), restoring the path $s\to v_2\to v_3\to v_4\to t$. To ensure that active entries exist when $\sigma(c)$ is small, each row must combine global vector neighbors with nearby geographic neighbors in one vector-sorted list.
\end{example}

\noindent\textbf{Why 1D Range and Per-Condition Indexes Fail.}
\label{sec:no-per-condition}
Range-filtered graphs~\cite{serf,window,urng,rnsg,Xu_2024} sort all objects into one global sequence $a_1\le\cdots\le a_n$ so that every query interval is contiguous. Two geometric properties rule out this strategy for 2D disks (proof of Theorem~\ref{thm:range-complexity} in \apx{app:range-complexity}):

\begin{theorem}[Disk Range Complexity]
\label{thm:range-complexity}
The family $\mathcal{A}(\dataset)=\{\Aset(c):c=(g,r),\ g\in\mathbb{S}^2,r\ge 0\}$ of distinct disk active sets satisfies $|\mathcal{A}(\dataset)| \le \sum_{i=0}^{4}\binom{n}{i} = O(n^4)$, and reaches $|\mathcal{A}(\dataset)| = \Omega(n^2)$ on worst-case $n$-object collections.
\end{theorem}

\begin{proposition}[Absence of a Universal Contiguous Order]
\label{prop:no-order}
Place three objects $V=\{1,2,3\}$ at the vertices of a small equilateral spherical triangle on $\mathbb{S}^2$. A disk centered at the midpoint of any pair $\{i,j\}$ with radius $\dgeo(s_i,s_j)/2$ admits only $\{i,j\}$; because the first and last elements in any permutation of $V$ are non-contiguous, no linear order makes $\Aset(c)$ contiguous for all disks $c$.
\end{proposition}

%% file: sections/03_recovery_graph.tex
\section{Ranked Active-Graph Recovery Theory}
\label{sec:theory}

Because 2D disks induce $\Omega(n^2)$ active sets and admit no contiguous order, we serve all disks from a \emph{single} static graph $G=(V,E)$ of bounded degree $B$ (complete proofs appear in \apx{app:proofs}).

\subsection{Active-Order Invariance and Recovery Graph}
\label{sec:recovery-model}

\noindent\textbf{Active-Order Invariance.}
Our first observation is that \emph{filtering a vector-sorted neighbor row preserves the relative vector-distance order of the surviving entries}. Suppose vertex $u$ stores its outgoing neighbors in a row $\rowu(u)=(v_1,\ldots,v_B)$ ordered strictly by distance $\dvec(x_u,x_v)$ (with ties broken by ID $v$, denoted $v\prec_u w$). Under disk $c=(g,r)$, masking out-of-disk entries leaves the active subsequence $C_u(c) = \rowu(u)\cap\Aset(c) = (v_i\in\rowu(u) : v_i\in\Aset(c))$, which remains strictly sorted under $\prec_u$. Let $\operatorname{rank}_{u,c}(v)\in\{1,\ldots,|C_u(c)|\}$ denote the 1-based ordinal rank of $v\in C_u(c)$ within $C_u(c)$.

\begin{example}
\label{ex:row-invariance}
Suppose $\rowu(u)=[a,b,c,d,e,f]$ is sorted by $\dvec(x_u,\cdot)$. Under disk $c_q$, if $\{a,c,d,e\}$ are active and $\{b,f\}$ are inactive (Figure~\ref{fig:fka-rule}), filtering leaves $C_u(c_q)=[a,c,d,e]$ with active ranks $1,2,3,4$. For budget $\rk=3$, a left-to-right scan skips $b$ and stops at the third active entry $d$ without reading $e$ or $f$, recovering $[a,c,d]$.
\end{example}

\begin{definition}[Complete Active Projection and Recovery Graph]
\label{def:recovery-graph}
Given vector-sorted rows $\{\rowu(u)\}_{u\in V}$, disk $c$, and integer \emph{recovery budget} $\rk\ge 1$, the \emph{recovery graph} $\RecG_{\rk}(c) = (\Aset(c), \{(u,v) : u\in\Aset(c), v\in C_u(c), \operatorname{rank}_{u,c}(v)\le \rk\})$ retains the first $\rk$ active row entries at each $u\in\Aset(c)$, and $\RecG_{\infty}(c)=G[\Aset(c)]$ is the \emph{complete active projection}.
\end{definition}

\begin{figure}[t]
  \centering
  \includegraphics[width=\columnwidth]{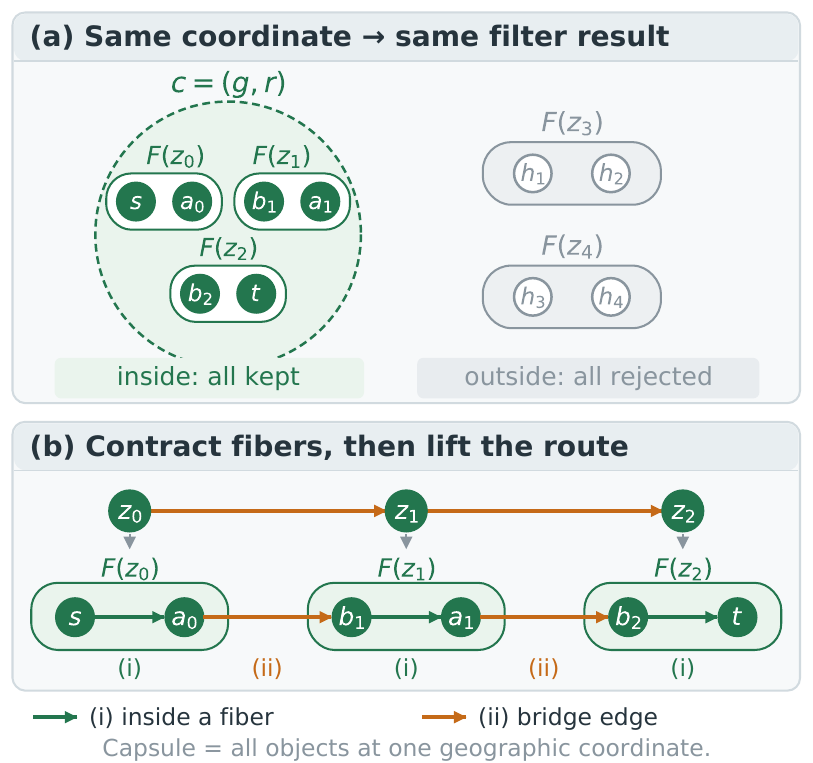}
  \caption{Two-scale route decomposition (Lemma~\ref{lem:fiber-closure}, Theorem~\ref{thm:factorization}). \emph{(a)}~Disk $c=(g,r)$ admits fibers $F(z_0),F(z_1),F(z_2)$ in full while rejecting $F(z_3),F(z_4)$ in full. \emph{(b)}~Contracting active fibers yields quotient route $z_0\to z_1\to z_2$, which lifts to object route $s\to a_0\to b_1\to a_1\to b_2\to t$ alternating intra-fiber paths (green, (i)) and inter-coordinate bridges (orange, (ii); Example~\ref{ex:factorization}).}
  \Description{Three whole fibers survive inside the disk while two whole fibers outside are rejected. Contracting active fibers gives the quotient route z0 to z1 to z2, which lifts to s, a0, b1, a1, b2, t.}
  \label{fig:factorization}
\end{figure}

\begin{figure*}[t]
  \centering
  \includegraphics[width=0.82\textwidth]{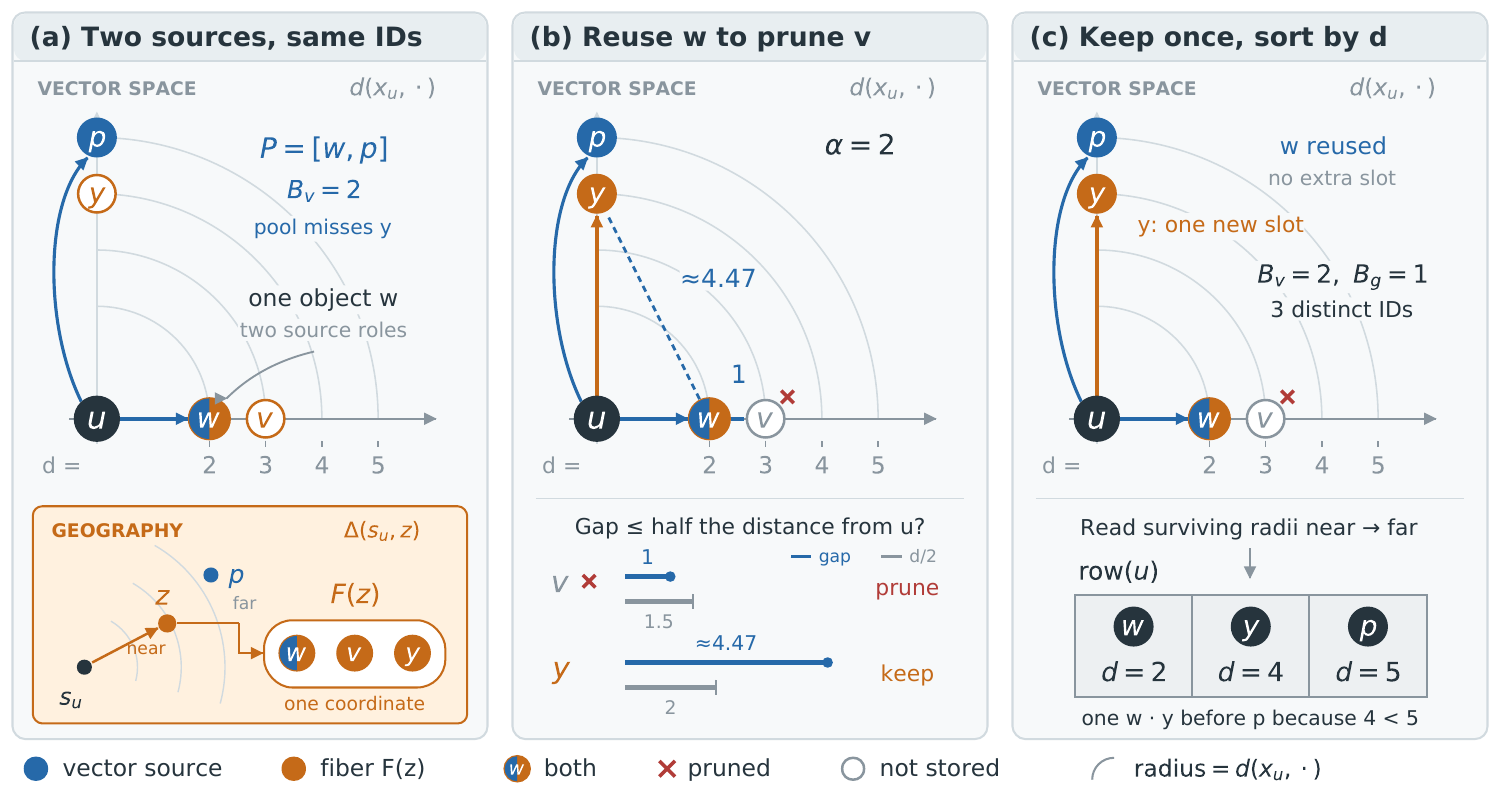}
  \caption{Dual-source row assembly in \method (Algorithm~\ref{alg:row-assembly}) on vector-distance shells $d=\dvec(x_u,\cdot)\in\{2,3,4,5\}$. \emph{(a)~Two sources, same IDs:} Stage~1 selects $P=[w,p]$ ($B_v=2$; pool misses $y$), while Stage~2 reaches nearest fiber $F(z)=\{w,v,y\}$. \emph{(b)~Reuse $w$ to prune $v$:} Inside $F(z)$ ($\alpha=2$), $w$ prunes $v$ ($\text{gap}=1\le 1.5=d/2$) while $y$ survives ($\text{gap}\approx 4.47>2=d/2$). \emph{(c)~Keep once, sort by $d$:} Reusing $w\in P$ takes no extra slot while $y$ takes $B_g=1$; reading surviving radii near-to-far yields $\rowu(u)=[w,y,p]$.}
  \Description{Three cards share concentric vector-distance arcs d = 2, 3, 4, 5 around u. Card (a) shows vector candidates P=[w,p] and nearest coordinate fiber F(z)={w,v,y}. Card (b) compares inter-candidate gap against d/2 at alpha=2, pruning v (1 <= 1.5) and keeping y (4.47 > 2). Card (c) reuses w without an extra slot and reads surviving radii near-to-far to produce row(u)=[w,y,p].}
  \label{fig:row-assembly}
\end{figure*}

\begin{theorem}[Subgraph Monotonicity and Cut-Rank Threshold]
\label{thm:cut-budget}
For any graph $G=(V,E)$ of row width $B$, disk $c$, and budget $\rk\ge 1$:
(a)~$\RecG_{\rk}(c)$ is a spanning subgraph of $\RecG_{\infty}(c)$ with out-degree at most $\min\{\rk,B\}$;
(b)~$\RecG_{\rk}(c)\subseteq\RecG_{\rk'}(c)$ for $\rk\le\rk'$ and $\RecG_B(c)=\RecG_{\infty}(c)$; and
(c)~for $|\Aset(c)|\ge 2$, $\RecG_{\rk}(c)$ is strongly connected if and only if $\rk\ge\rk_{\mathrm{SC}}(c)$, where the \emph{cut rank} is
\begin{equation}
\label{eq:cut-rank}
  \rk_{\mathrm{SC}}(c) =
  \max_{\varnothing\ne S\subsetneq\Aset(c)}\;
  \min_{u\in S,\, v\in C_u(c)\setminus S}
  \operatorname{rank}_{u,c}(v)
\end{equation}
(with inner minimum $\infty$ if no edge of $\RecG_{\infty}(c)$ leaves $S$). In~\eqref{eq:cut-rank}, any subset $S\subsetneq\Aset(c)$ represents a potential isolated component: the inner minimum is the earliest active rank at which some vertex $u\in S$ has an edge leaving $S$, and the outer maximum is the worst-case cut across $\Aset(c)$ (proof in \apx{app:main-proofs}).
\end{theorem}

\subsection{Coordinate Fibers and Route Decomposition}
\label{sec:fiber-structure}

\noindent\textbf{Coordinate Fibers and Closure.}
Real geotagged datasets differ widely in coordinate multiplicity. In POI catalogs (\textsc{Overture-1M}), almost all records have distinct coordinates ($M\approx n$); in social-media archives (\textsc{Yfcc-Clip}), $64.6\%$ of coordinates are duplicates and up to $5{,}726$ photos share a single landmark location~\cite{Antoniou_2010}. To model the entire range $1\le M\le n$, we group all objects at each distinct coordinate $z\in\mathcal{S}$ into a \emph{coordinate fiber} $F(z)=\{v\in V:s_v=z\}$. Objects in the same fiber $F(z)$ have $\dgeo=0$ and differ only in vector space $\mathcal{X}$, while the $M$ fibers $\{z_1,\ldots,z_M\}$ form distinct spatial sites on $\mathcal{S}$.

\begin{lemma}[Coordinate-Fiber Closure]
\label{lem:fiber-closure}
For every disk $c=(g,r)$ and coordinate $z$, $F(z)\cap\Aset(c)\ne\varnothing \iff F(z)\subseteq\Aset(c)$. Because all $v\in F(z)$ have the same distance $\dgeo(z,g)$ to $g$, a disk never splits a fiber (Figure~\ref{fig:factorization}(a)), and the induced fiber subgraph $G[F(z)]$ is identical under every disk that admits $z$.
\end{lemma}

\begin{theorem}[Two-Scale Route Decomposition]
\label{thm:factorization}
For any directed graph $G=(V,E)$ and disk $c$, $G[\Aset(c)]$ is strongly connected if:
(i)~\textbf{(Intra-fiber connectivity)} $G[F(z)]$ is strongly connected for every active fiber $F(z)\subseteq\Aset(c)$; and
(ii)~\textbf{(Coordinate-quotient connectivity)} the quotient graph formed by contracting each active fiber $F(z)\subseteq\Aset(c)$ into a single site vertex $z$ is strongly connected.
\end{theorem}

\begin{example}
\label{ex:factorization}
In Figure~\ref{fig:factorization}(b), the quotient path $z_0\to z_1\to z_2$ across active sites is supported by inter-coordinate bridge edges $a_0\to b_1$ and $a_1\to b_2$ (Condition~(ii)). Condition~(i) connects these bridges through intra-fiber paths $s\to a_0$ in $F(z_0)$, $b_1\to a_1$ in $F(z_1)$, and $b_2\to t$ in $F(z_2)$, forming an end-to-end path in $G[\Aset(c)]$. Neither condition scales with the number of possible query disks: Condition~(i) involves only $M$ static fiber subgraphs fixed at index construction time, and Condition~(ii) depends only on the spatial sites inside $\Omega_c$.
\end{example}

\subsection{Polylogarithmic Guarantees}
\label{sec:obligations}

\noindent\textbf{Coordinate Tie-Breaking and Two-Scale Selection.}
To maintain connectivity on small disks, each vertex $u$ allocates a quota of $B_g$ edges to geographically nearby candidates. If $u$ simply linked to its $B_g$ geographically closest objects, a nearby landmark $z$ with thousands of co-located photos would fill all $B_g$ slots with ties at $z$ and leave no edges to surrounding sites, violating Condition~(ii) of Theorem~\ref{thm:factorization}; conversely, keeping only one arbitrary object per site would disconnect multi-object fibers and violate Condition~(i).

To satisfy both conditions, we model $\dataset$ by a regular marked point process~\cite{DaleyVereJones_2003,Penrose_2003} (\apx{app:fiber-connectivity}): $M$ distinct sites $Z_1,\ldots,Z_M$ are drawn on a compact, connected region $\mathcal{S}\subseteq\mathbb{S}^2$ with bounded density $0<p_{s,-}\le p_s(z)\le p_{s,+}$, and each site $Z_i$ hosts a fiber $F(Z_i)$ of size $f_i\ge 1$ ($\sum_{i=1}^M f_i=n$) whose vectors $X_v\sim\nu_{Z_i,\Theta_i}$ come from a geodesic support of intrinsic dimension $d_{\mathcal{X}}>0$, diameter $D_x$, and doubling constant $C_D$. Each vertex $u$ selects geographic edges using two rules:
(1)~\emph{Intra-fiber $\alpha$-witness rule:} within $F(z)$, $u$ scans up to $K_f$ candidates $v\in F(z)\setminus\{u\}$ in ascending $\dvec(x_u,x_v)$ and retains $v$ unless an earlier retained $w\in F(z)$ satisfies $\dvec(x_v,x_w)\le\dvec(x_u,x_v)/\alpha$ ($\alpha>1$).
(2)~\emph{Nearest-first coordinate stream:} $u$ visits sites $z\ne s_u$ in nondecreasing $\dgeo(s_u,z)$ and applies Rule~(1) inside each $F(z)$ until $B_g$ new edges are selected, with cumulative charge $L_u(\rho)=\sum_{z:\dgeo(s_u,z)\le\rho}\chi_u(z)$ within radius $\rho$.

\noindent\textbf{Bounding Per-Site Quota Consumption.}
Partition $\mathcal{X}\setminus\{x_u\}$ into dyadic shells $[2^{-k}D_x, 2^{-(k-1)}D_x]$ from minimum separation $\varepsilon_\alpha=\Theta((\delta/n^2)^{1/d_{\mathcal{X}}})$ to $D_x$. Because any two $\alpha$-pruned survivors in the same shell are separated by $>1/(2\alpha)$ of the shell radius, doubling packing limits each shell to $\le P_\alpha=C_D^{\lceil\log_2(4\alpha)\rceil}=O(1)$ survivors; across all shells, with probability $\ge 1-\delta/3$, \emph{any site $z$ of arbitrary size $f_z$ consumes at most}
\begin{equation}
\label{eq:alpha-charge}
  a_\alpha = \min\!\left\{K_f,\; P_\alpha\!\left(1+\left\lceil\log_2(D_x/\varepsilon_\alpha)\right\rceil\right)\right\} = O\bigl(\log(nM/\delta)\bigr)
\end{equation}
slots of $B_g$. Meanwhile, RGG theory~\cite{Penrose_1999,Penrose_2003} shows that linking active sites within radius $\rho_c = C_s\sqrt{\lambda_2(\Omega_c)\log(3M/\delta)/(M\pi_c)}$ ($\pi_c=\int_{\Omega_c}p_s\,d\lambda_2$) connects the coordinate quotient across $\Omega_c$, with at most $h_c=\lceil C_h\log(3M/\delta)\rceil$ sites per $\rho_c$-ball (\apx{app:site-backbone}).

\begin{theorem}[Two-Scale and End-to-End Strong Connectivity]
\label{thm:fiber-connectivity}
\label{thm:quotient-connectivity}
\label{thm:end-to-end}
Under the regular marked-site model (\apx{app:fiber-connectivity}), for any disk $c$ with $M\pi_c\ge\log(6/\delta)$, setting quotas
\begin{equation}
\label{eq:quota-scaling}
  B_g \ge a_\alpha(h_c+1)=\Theta\bigl(\log^2(nM/\delta)\bigr),
  \;
  K_f \ge C_f\log(3nM/\delta)
\end{equation}
guarantees that: (a)~with probability $\ge 1-\delta/3$ across all $M$ fibers, every $G[F(z)]$ is strongly connected with $\alpha$-contracting local paths; (b)~with probability $\ge 1-2\delta/3$, the coordinate quotient of $G[\Aset(c)]$ is strongly connected and every $u\in\Aset(c)$ retains a representative from every active site within $\rho_c$; and (c)~consequently, $\RecG_{\infty}(c)=G[\Aset(c)]$ is strongly connected ($\rk_{\mathrm{SC}}(c)\le B$) with probability $\ge 1-\delta$ (and across $Q$ disks upon replacing $\delta$ by $\delta/Q$).
\end{theorem}

\begin{corollary}[$\alpha$-Shortcut Reachability on Compact Disks]
\label{cor:navigability}
On the event of Theorem~\ref{thm:end-to-end}, if $\operatorname{diam}(\Omega_c)\le\rho_c$ and $K_f\ge |F(z)|-1$ for $F(z)\subseteq\Aset(c)$, then for every $u,v\in\Aset(c)$ with $\dvec(x_u,x_v)>\varepsilon_\alpha$, $u$ has an active out-neighbor $w\in F(s_v)$ satisfying $\dvec(x_w,x_v)\le\dvec(x_u,x_v)/\alpha$. Greedy routing therefore reaches $v$ in at most $\lceil\log_\alpha(D_x/\varepsilon_\alpha)\rceil$ hops, and any local minimum $u$ for query $x_q$ satisfies $\dvec(x_q,x_u)\le\frac{\alpha+1}{\alpha-1}\min_{v\in\Aset(c)}\dvec(x_q,x_v)$.
\end{corollary}

Because $B_g$ and $K_f$ grow only polylogarithmically in $n$ and $M$ and do not depend on $\max_i f_i$, a single width ($B=384$) suffices across scales: when $M=n$ (\textsc{Overture-1M}), Condition~(i) is trivial ($\chi_u(z)=1$) and $B_g$ builds a $B_g$-NN spatial graph; when $M\ll n$ (\textsc{Yfcc-Clip}), intra-fiber $\alpha$-pruning bounds per-site quota usage and keeps selective disks in the compact regime $\operatorname{diam}(\Omega_c)\le\rho_c$.

%% file: sections/04_reach_index.tex
\section{The Proposed \method Index}
\label{sec:index}

To realize the guarantees of Section~\ref{sec:theory} in a single fixed-width graph, \method stores one row $\rowu(u)$ of width $B=B_v+B_g$ for each vertex $u\in V$. Each row combines a \emph{global-vector source} of quota $B_v$ (which provides long-range vector shortcuts on broad disks where $\rho_c<\operatorname{diam}(\Omega_c)$) with a \emph{fiber-aware geographic stream} of quota $B_g$ (which provides local connectivity on selective disks under Theorem~\ref{thm:end-to-end}), and sorts all retained neighbors by $\dvec(x_u,\cdot)$.

\subsection{Dual-Source Row Construction}
\label{sec:vecpool}
\label{sec:assembly}

Algorithm~\ref{alg:row-assembly} constructs $\rowu(u)$ for each vertex $u\in V$ in three stages (Figure~\ref{fig:row-assembly}).
In \textbf{Stage 1 (Lines 1--4)}, NN-Descent~\cite{nndescent} builds an approximate $K_{\mathrm{vec}}$-NN candidate pool $C_v(u)$ in $\mathcal{X}$; after deduplicating $C_v(u)$ and sorting it by exact \texttt{float32} distance under $\prec_u$, lines~2--4 apply $\alpha$-pruning and back-fill up to quota $B_v$.
In \textbf{Stage 2 (Lines 5--11)}, each coordinate $s=(\lambda,\phi)$ is mapped to its 3D unit vector $\mathbf{u}_s=(\cos\phi\cos\lambda,\allowbreak \cos\phi\sin\lambda,\sin\phi)$. Because squared chord distance $\|\mathbf{u}_s-\mathbf{u}_t\|_2^2=2-2\langle\mathbf{u}_s,\mathbf{u}_t\rangle$ is strictly monotone in $\dgeo(s,t)$, grouping identical unit vectors forms the fibers $\{F(z_1),\ldots,F(z_M)\}$, and a 3D kd-tree~\cite{Bentley_1975} streams fibers in exact great-circle order from $s_u$. For each fiber $F(z)$, line~6 resets the co-located witness list $W\gets[\,]$, and lines~7--11 scan up to $C_{\mathrm{tie}}$ members in ascending $\dvec(x_u,\cdot)$: if an earlier witness $w\in W$ satisfies $\dvec(x_v,x_w)\le\dvec(x_u,x_v)/\alpha$, candidate $v$ is pruned (line~9); otherwise $v$ enters $W$ (line~10), and if $v\notin P$, it is appended to $P$ and increments $g$ (line~11). Thus a vector survivor $w\in P$ still serves as a co-located witness in $W$ at zero geographic quota cost ($+0$), with $K_f=\min\{B_g,C_{\mathrm{tie}}\}$.
In \textbf{Stage 3 (Lines 12--14)}, lines~12--13 fill remaining slots from $C_v(u)$ or sentinel $\bot=\texttt{0xFFFFFFFF}$, and line~14 sorts $P$ by $(\dvec(x_u,x_v),v)$, removing source labels.

\begin{example}
\label{ex:row-assembly}
Figure~\ref{fig:row-assembly} traces Algorithm~\ref{alg:row-assembly} for $B_v=2$, $B_g=1$, and $\alpha=2$, with $x_u=(0,0)$ and shells $d\in\{2,3,4,5\}$ containing $x_w=(2,0)$, $x_v=(3,0)$, $x_y=(0,4)\in F(z)$ and $x_p=(0,5)\notin F(z)$. In panel~(a), Stage~1 selects $P=[w,p]$ ($C_v(u)$ misses $y$), while Stage~2 visits nearest site $z$ with $F(z)=\{w,v,y\}$. In panel~(b), scanning $F(z)$ in order $d=2,3,4$ places $w$ into $W$ without using a new geographic slot since $w\in P$; next, $v$ is pruned by $w$ ($\dvec(x_v,x_w)=1\le 1.5=d/2$), whereas $y$ survives ($\dvec(x_y,x_w)=\sqrt{20}\approx 4.47>2=d/2$) and takes the single slot $g=1=B_g$. In panel~(c), sorting the retained neighbors $\{w,y,p\}$ from inner to outer shell ($d=2,4,5$) gives $\rowu(u)=[w,y,p]$.
\end{example}

\begin{figure}[t]
  \centering
  \includegraphics[width=\columnwidth]{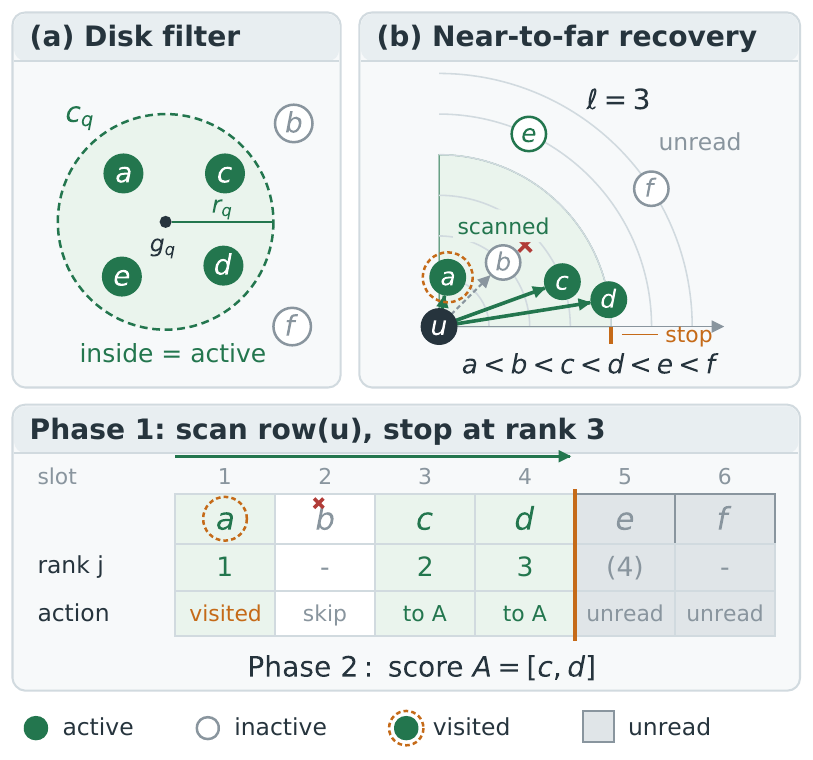}
  \caption{Two-phase first-$\rk$-active expansion ($\rk=3$) at vertex $u$ (Example~\ref{ex:row-invariance}, Algorithm~\ref{alg:expand}). \emph{(a)}~Disk $c_q=(g_q,r_q)$ activates $\{a,c,d,e\}$ and deactivates $\{b,f\}$. \emph{(b)}~Concentric arcs order $u$'s neighbors by vector distance ($a<b<c<d<e<f$): skipping out-of-disk $b$ recovers the first $\rk=3$ active neighbors $[a,c,d]$ while leaving outer arcs $e,f$ unread. \emph{Bottom:}~Phase~1 scans slots $1\ldots 4$ ($a\in\mathcal{V}$ takes rank~1, $b$ is skipped, $c,d$ take ranks~2, 3 into $A=[c,d]$) and halts at $j=3$; Phase~2 scores only $A=[c,d]$.}
  \Description{The geographic disk selects a, c, d, e. Concentric arcs show vector distance order a < b < c < d < e < f. The bottom slot table shows Phase 1 stopping after slot 4 and Phase 2 scoring c and d.}
  \label{fig:fka-rule}
\end{figure}

\begin{algorithm}[t]
\small
\caption{\textsc{AssembleRow}$(u)$: Dual-source row construction}
\label{alg:row-assembly}
\KwIn{pool $C_v(u)$ sorted by $\dvec(x_u,\cdot)$; fibers $F(z)$ streamed by $\dgeo(s_u,z)$}
\KwOut{fixed-width row $\rowu(u)$ of width $B=B_v+B_g$ sorted by $\prec_u$}
$P\gets[\,]$;\quad $g\gets 0$\;
\ForEach{$v\in C_v(u)$ \textbf{while} $|P|<B_v$}{
  \lIf{$\dvec(x_v,x_w)>\dvec(x_u,x_v)/\alpha\;\forall w\in P$}{append $v$ to $P$}
}
back-fill $P$ up to $B_v$ with closest unused $v\in C_v(u)$\;
\ForEach{coordinate $z$ in nondecreasing $\dgeo(s_u,z)$ \textbf{while} $g<B_g$}{
  $W\gets[\,]$;\quad $t\gets 0$\;
  \ForEach{$v\in F(z)\setminus\{u\}$ in vector order \textbf{while} $g<B_g$ \textbf{and} $t<C_{\mathrm{tie}}$}{
    $t\gets t+1$\;
    \lIf{$\exists w\in W:\dvec(x_v,x_w)\le \dvec(x_u,x_v)/\alpha$}{\textbf{continue}}
    append $v$ to $W$\;
    \lIf{$v\notin P$}{append $v$ to $P$;\quad $g\gets g+1$}
  }
}
back-fill $P$ up to $B$ with unused $v\in C_v(u)$, then sentinels $\bot$\;
\Return $P$ sorted by $\bigl(\dvec(x_u,x_v),v\bigr)$, with sentinels $\bot$ last\;
\end{algorithm}

\begin{proposition}[Row Construction Contract]
\label{prop:assembly-contract}
For $B_v+B_g\le B$, Algorithm~\ref{alg:row-assembly} guarantees that: (i)~every candidate $v\in F(z)\setminus\{u\}$ examined in lines~7--11 is either stored in $\rowu(u)$ or witnessed by some $w\in\rowu(u)\cap F(z)$ with $\dvec(x_v,x_w)\le\dvec(x_u,x_v)/\alpha$; (ii)~if $L_u(\rho)\le B_g$, every site $z\ne s_u$ with $\dgeo(s_u,z)\le\rho$ has $\rowu(u)\cap F(z)\ne\varnothing$; and (iii)~$\rowu(u)$ contains at most $B$ valid IDs ordered by $\prec_u$ ($4nB$ bytes in total).
\end{proposition}

\subsection{Online Query Processing and Spatial Entry}
\label{sec:search}
\label{sec:entry}
\label{sec:autotune}

\noindent\textbf{Two-Phase First-$\rk$-Active Expansion (Algorithm~\ref{alg:expand}).}
At query time, \method runs beam search with list size $L$ from active entry seeds and expands each visited vertex $u$ in two phases (Figure~\ref{fig:fka-rule}).
In \emph{Phase 1 (Lines 1--7)}, the algorithm scans $\rowu(u)$ sequentially, testing each neighbor $v$ against disk $c_q$ via a 32-bit bounding box and a 3D dot product $\langle\mathbf{u}_{s_v},\mathbf{u}_{g_q}\rangle\ge\cos(r_q/R_{\oplus})$ (\apx{app:stream-numerics}). Whenever $v\in\Aset(c_q)$, the active-rank counter increments ($j\gets j+1$), unvisited $v\notin\mathcal{V}$ enters buffer $A$, and the scan stops once $j=\rk$.
In \emph{Phase 2 (Lines 8--10)}, $D$-dimensional Euclidean distances $\dvec(x_q,x_v)$ are computed only for $v\in A$ (at most $\min\{\rk,B\}$ vector evaluations vs.\ at most $B$ scalar checks in Phase~1). Because $j$ increments on every active entry (including visited $a\in\mathcal{V}$ in Figure~\ref{fig:fka-rule}), beam search traverses the exact static subgraph $\RecG_{\rk}(c_q)$, where $\rk_{\mathrm{exact}}(c)=\rk_{\mathrm{SC}}(c)$, $\rk_{\mathrm{safe}}=B$, and $\Pr\{\rk_{\mathrm{SC}}(c)\le B\}\ge 1-\delta$.

\begin{algorithm}[t]
\small
\caption{\textsc{Expand}$(u,q,\rk)$: Two-phase first-$\rk$-active expansion}
\label{alg:expand}
\KwIn{$u\in\Aset(c_q)$; query $q=(x_q,c_q,k)$; budget $\rk$; visited set $\mathcal{V}$}
\KwOut{scored active neighbor list $D$ for the search beam}
$A\gets[\,]$;\quad $j\gets 0$;\quad $D\gets[\,]$\;
\ForEach{$v\in\rowu(u)$ in stored vector-distance order}{
  \lIf{$v=\bot$}{\textbf{break}}
  \lIf{$\dgeo(s_v,g_q)>r_q$}{\textbf{continue}}
  $j\gets j+1$\tcp*{active rank $\operatorname{rank}_{u,c_q}(v)$}
  \lIf{$v\notin\mathcal{V}$}{append $v$ to $A$}
  \lIf{$j=\rk$}{\textbf{break}}
}
\ForEach{$v\in A$}{
  append $\bigl(v,\dvec(x_q,x_v)\bigr)$ to $D$;\quad $\mathcal{V}\gets\mathcal{V}\cup\{v\}$\;
}
\Return $D$\;
\end{algorithm}

\begin{figure*}[!t]
  \centering
  \includegraphics[width=\textwidth]{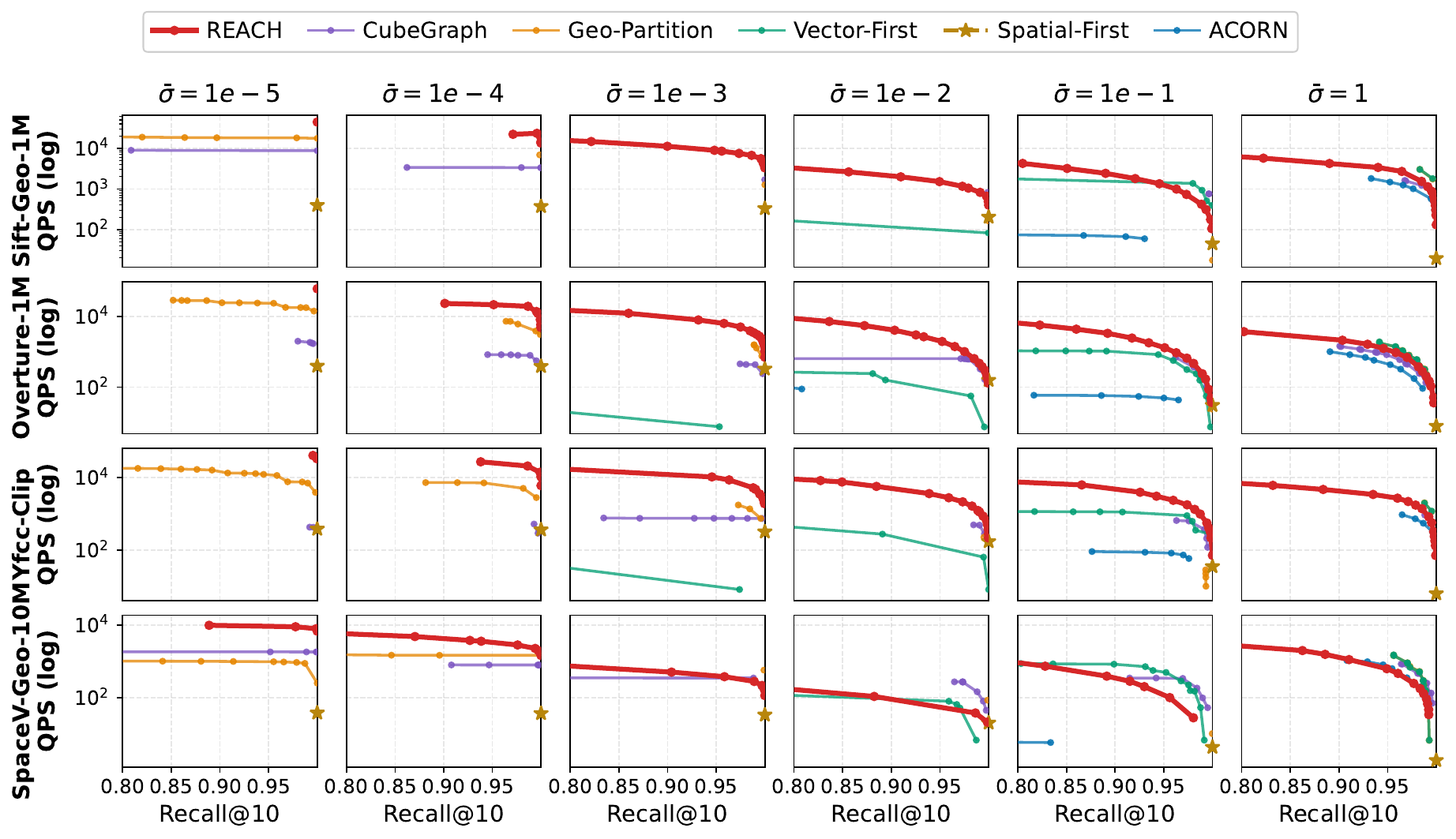}
  \caption{Single-threaded Recall@10--QPS curves across four datasets (rows) and six selectivities $\bar{\sigma}\in[10^{-5},1]$ (columns). All six panels in each row use the exact same frozen \method index.}
  \Description{A four-by-six matrix of recall-versus-QPS plots compares REACH with five baselines across four datasets and six selectivities.}
  \label{fig:main-curves}
\end{figure*}

\noindent\textbf{Practical Selection of $\rk$, Spatial Seeding, and Dynamic Updates.}
Because $\rk_{\mathrm{SC}}(c)$ has a median of $48$ on selective queries (Exp-3), a fixed default $\rk=64$ (or $\rk=32$ for $\bar{\sigma}\le 10^{-3}$) stays within $0.0032$ of the full-row ($\rk=\infty$) recall ceiling across all $20$ filtered cells in our benchmark; alternatively, using the empirical active fraction $\hat{\sigma}(c_q)$ observed during seeding, the rule $\rk(\hat{\sigma})=\min\{B,\lceil\rk_0+B_v\hat{\sigma}\rceil\}$ ($\rk_0\in[32,64]$) sets $\rk$ automatically. To seed the beam, a $0.5^\circ\times 0.5^\circ$ grid directory~\cite{Gaede_1998,Samet_1984} scans Chebyshev rings around $g_q$ to collect $\tau=16$ active seeds in $1.2\,\mu\mathrm{s}$ using $<0.4\%$ of graph memory (\apx{app:cell-index}; graph-descent seeding is compared in Exp-5). For dynamic updates, inserting $u_{\mathrm{new}}$ updates its grid cell and 3D tree in $O(\log M)$ time, builds $\rowu(u_{\mathrm{new}})$ through one beam search and a local coordinate stream, and inserts $u_{\mathrm{new}}$ into reverse neighbors' sorted rows; deletions set a tombstone bit skipped in Phase~1.

%% file: sections/05_experiments.tex
\section{Experiments}
\label{sec:experiments}

\begin{figure*}[!t]
  \centering
  \includegraphics[width=0.76\textwidth]{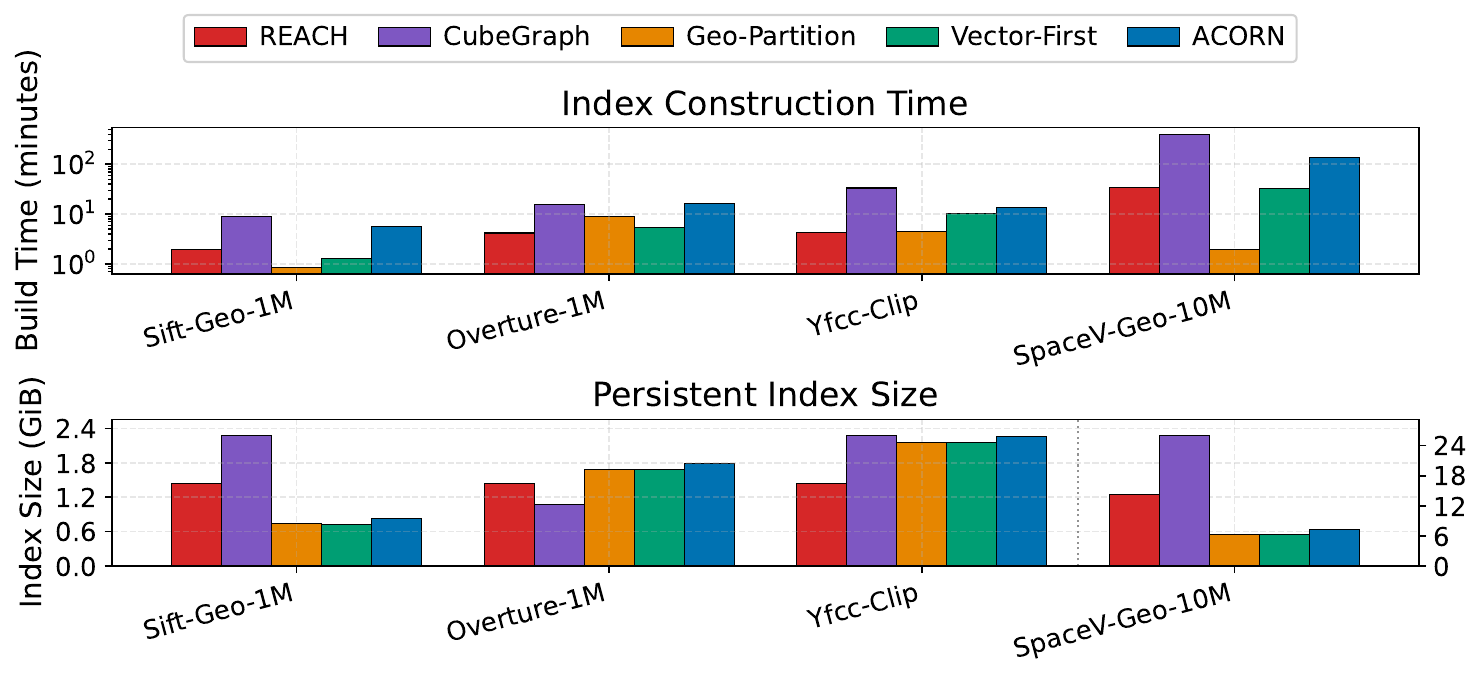}
  \caption{Index construction time (top, log scale) and persistent index size on disk (bottom; right axis for \textsc{SpaceV-Geo-10M}).}
  \Description{Grouped bars compare construction time and persistent index size for REACH and the baselines on four datasets.}
  \label{fig:build-resources}
\end{figure*}

\begin{figure*}[!t]
  \centering
  \includegraphics[width=\textwidth]{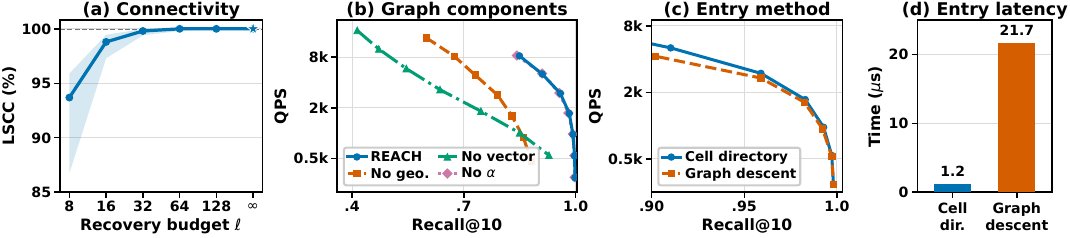}
  \caption{Structural connectivity check (\emph{a}, $\bar{\sigma}=10^{-3}$) and source/seeding ablations (\emph{b--d}, $\bar{\sigma}=10^{-2}, \rk=32$) on \textsc{Yfcc-Clip}.}
  \Description{Panel (a) shows largest-SCC coverage vs recovery budget; panels (b)-(d) compare row-source ablations and entry seeding.}
  \label{fig:core-entry-ablation}
\end{figure*}

\subsection{Experimental Setup}
\label{sec:setup}

\begin{table}[t]
  \centering
  \caption{Benchmark datasets ($1{,}000$ base-disjoint queries per cell).}
  \label{tab:datasets}
  \footnotesize
  \setlength{\tabcolsep}{1.5pt}
  \begin{tabular}{@{}lrrll@{}}
    \toprule
    \textbf{Dataset} & $n$ & $D$ & \textbf{Embeddings (\texttt{f32}, L2)} & \textbf{Coordinates}\\
    \midrule
    \textsc{Overture-1M} & $10^6$ & $384$ & POI \texttt{MiniLM-L6-v2} & real; $98.9\%$ distinct\\
    \textsc{Yfcc-Clip} & $9.98\times 10^5$ & $512$ & CLIP ViT-B/32~\cite{yfcc100m,clip} & real; $35.4\%$ distinct\\
    \textsc{Sift-Geo-1M} & $10^6$ & $128$ & official SIFT1M~\cite{sift1m} & synth.; $100\%$ distinct\\
    \textsc{SpaceV-Geo-10M} & $10^7$ & $100$ & official SpaceV~\cite{pmlr-v176-simhadri22a} & synth.; $100\%$ distinct\\
    \bottomrule
  \end{tabular}
\end{table}

\noindent\textbf{Datasets, Workloads, and Baselines.}
Table~\ref{tab:datasets} lists the four datasets: \textsc{Overture-1M} (US POIs from Overture Maps~2026-06-17.0, max fiber $41$), \textsc{Yfcc-Clip} (Web photos with $64.6\%$ repeated coordinates, max fiber $5{,}726$), and \textsc{Sift-Geo-1M} and \textsc{SpaceV-Geo-10M} (paired with uniform US coordinates independent of the vectors). We test six selectivities $\bar{\sigma}\in\{10^{-5},\allowbreak 10^{-4},\allowbreak 10^{-3},\allowbreak 10^{-2},\allowbreak 10^{-1},1\}$: $\bar{\sigma}\in[10^{-5},10^{-3}]$ ($10$--$10^4$ objects) represents street- or walking-scale radii, $\bar{\sigma}\in[10^{-2},10^{-1}]$ ($10^4$--$10^6$ objects) city- or metropolitan-scale radii, and $\bar{\sigma}=1$ unfiltered search, setting $r_q$ to the $m_{\mathrm{tar}}=\lceil\bar{\sigma}n\rceil$-th spatial-neighbor distance from $g_q$. We report $\mathrm{Recall}@10$ and single-threaded $\mathrm{QPS}$ over $1{,}000$ queries per cell against five baselines: (1)~\emph{Spatial-First} (exact spatial filter followed by vector scan); (2)~\emph{Vector-First} (global HNSW~\cite{hnsw} with post-filtering); (3)~\emph{Geo-Partition} (grid of cell-local HNSW indexes, tuning widths $\{1^\circ,0.5^\circ,0.25^\circ\}$ per dataset); (4)~\emph{CubeGraph}~\cite{https://doi.org/10.48550/arxiv.2604.06616} (hierarchical spatial-tree graph, tuning profiles $P_0,P_1,P_2$ per dataset); and (5)~\emph{ACORN}~\cite{acorn} ($M=32,\gamma=12,M_\beta=32$; \emph{Spatial-First} and \emph{Geo-Partition} also instantiate the primitives of Mesh~\cite{DBLP:journals/vldb/SongYCYXLC25}). Across all datasets, \method uses \emph{one fixed profile}: $B=384, B_v=128, B_g=256, K_{\mathrm{vec}}=192, \alpha=1.2, C_{\mathrm{tie}}=4096, \tau=16$, sweeping $\rk\in[8,\infty]$ and $L\in[16,1024]$ ($4096$ at $10^{-2}$ and $2048$ at $10^{-1}$ on $10^7$; \apx{app:repro}) on a 64-core AMD 3990X ($220$~GiB RAM, GCC~11.2 \texttt{-O3 -march=native}; \url{https://anonymous.4open.science/r/GeoSearch-new-1730/}).

\subsection{Experimental Results}
\label{sec:results}

\begin{figure}[!t]
  \centering
  \includegraphics[width=\columnwidth]{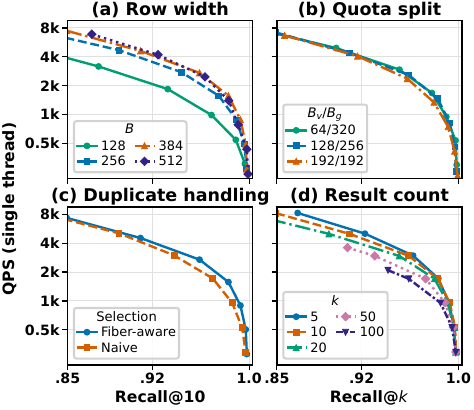}
  \caption{Sensitivity and robustness on \textsc{Yfcc-Clip} ($\bar{\sigma}=10^{-2}, \rk=32$): row width $B$, quota split $B_v/B_g$, fiber-aware $\alpha$-pruning vs.\ naive truncation, and result count $k\in[5,100]$.}
  \Description{Four recall-QPS plots vary row width, quota split, duplicate-coordinate selection, and requested result cardinality.}
  \label{fig:robustness}
\end{figure}

\noindent\textbf{Exp-1: End-to-End Query Throughput (Figure~\ref{fig:main-curves}).}
Across the $16$ selective cells ($\bar{\sigma}\in\{10^{-5},10^{-4},10^{-3},10^{-2}\}$), one \method index per dataset ranks fastest in $14$ of $16$ cells at $0.95$ Recall@10, achieving $2.49\times$ higher geometric-mean QPS than the best baseline in each cell ($6.33\times$ over CubeGraph, $3.56\times$ over Geo-Partition across selective cells where both reach $0.95$ recall).
\emph{(1)~Performance on selective disks:} At $\bar{\sigma}=10^{-5}$ ($0.95$ recall), \method reaches $60{,}760$~QPS (\textsc{Overture-1M}), $44{,}255$~QPS (\textsc{Sift-Geo-1M}), $41{,}191$~QPS (\textsc{Yfcc-Clip}), and $8{,}901$~QPS (\textsc{SpaceV-Geo-10M}), beating the best baseline per panel (Geo-Partition at $23{,}627, 18{,}029, 11{,}552$~QPS on $10^6$, and CubeGraph at $1{,}842$ vs.\ $988$~QPS on $10^7$) by $2.57\times, 2.45\times, 3.57\times, 4.83\times$ ($30.5\times, 95.2\times$ over CubeGraph on \textsc{Overture-1M}, \textsc{Yfcc-Clip}). At $\bar{\sigma}=10^{-4}, 10^{-3}$, \method reaches $21{,}495, 6{,}336$~QPS on \textsc{Overture-1M} ($2.94\times, 4.04\times$ over Geo-Partition; $26.0\times, 14.0\times$ over CubeGraph) and $21{,}046, 8{,}707$~QPS on \textsc{Yfcc-Clip} ($3.77\times, 4.92\times$ over Geo-Partition; $39.9\times, 11.5\times$ over CubeGraph).
\emph{(2)~Baseline limitations:} \emph{Vector-First} is fast at $\bar{\sigma}=1$ ($1{,}411$--$3{,}020$~QPS) but misses $0.95$ recall in all $12$ cells with $\bar{\sigma}\le 10^{-3}$ due to path erosion; \emph{Spatial-First} reaches $384$--$396$~QPS at $10^{-5}$ ($105\times$--$154\times$ slower) and $2.1$--$19.5$~QPS at $\bar{\sigma}=1$; \emph{ACORN} misses $0.95$ recall in $18/20$ filtered cells; and \emph{Geo-Partition} drops to $17.6$--$28.7$~QPS at $\bar{\sigma}=10^{-1}$ on $10^6$ ($10.5$~QPS on $10^7$), slower than \emph{Spatial-First} ($30.5$--$44.9$~QPS).
\emph{(3)~Clustered vs.\ uniform coordinates:} On both real datasets (\textsc{Overture-1M}, \textsc{Yfcc-Clip}), \method leads across \emph{all} filtered selectivities $\bar{\sigma}\in[10^{-5},10^{-1}]$ ($2{,}804, 2{,}384$~QPS at $10^{-2}, 10^{-1}$ on \textsc{Yfcc-Clip}, $5.60\times, 3.61\times$ over CubeGraph) and at $\bar{\sigma}=1$ on \textsc{Yfcc-Clip} ($2{,}810$~QPS). On synthetic \textsc{SpaceV-Geo-10M}, \method leads at $10^{-5},10^{-4}$, beats CubeGraph ($381$ vs.\ $351$~QPS) while trailing Geo-Partition ($586$~QPS) at $10^{-3}$, and trails at $10^{-2},10^{-1}$: by Corollary~\ref{cor:navigability}, with $M=10^7$ uniform sites ($10\times$ \textsc{Sift-Geo-1M}), a disk at $\bar{\sigma}\in[10^{-3},10^{-1}]$ covers $10^4$--$10^6$ sites, exceeding $\rho_c$ of $B_g=256$ while leaving $\bar{\sigma}B_v\approx 0.13$--$12.8$ active vector edges per row, whereas real datasets are spatially clustered.

\noindent\textbf{Exp-2: Build Time and Index Size (Figure~\ref{fig:build-resources}).}
\method builds in $117$~s (\textsc{Sift-Geo-1M}), $250$~s (\textsc{Overture-1M}), $252$~s (\textsc{Yfcc-Clip}), and $2{,}007$~s (\textsc{SpaceV-Geo-10M}), which is $4.6\times, 3.8\times, 7.9\times$, and $11.9\times$ faster than CubeGraph ($545$--$23{,}858$~s), $3.0\times$--$4.0\times$ faster than ACORN, and faster than Vector-First ($319, 613$~s) and Geo-Partition ($539, 274$~s) at $D\in\{384,512\}$. Its $4nB$-byte row table plus $O(n)$ grid ($<0.4\%$) uses $1.43$--$1.44$~GiB at $10^6$ and $14.34$~GiB at $10^7$, which is $37\%$--$45\%$ smaller than CubeGraph on \textsc{Sift-Geo-1M}, \textsc{Yfcc-Clip} ($2.28$~GiB), and \textsc{SpaceV-Geo-10M} ($26.08$~GiB; CubeGraph is $1.07$~GiB on \textsc{Overture-1M}), while FAISS stores the $4nD$-byte vectors ($1.68$--$2.26$~GiB).

\noindent\textbf{Exp-3: Direct Connectivity Verification (Figure~\ref{fig:core-entry-ablation}(a)).}
Over $1{,}000$ queries on \textsc{Yfcc-Clip} at $\bar{\sigma}=10^{-3}$, $\RecG_{\infty}(c)$ achieves $\mathrm{LSCC}=100\%$ on $98.4\%$ of queries ($>99\%$ on the rest; Theorem~\ref{thm:end-to-end}): active fibers are strongly connected for $99.998\%$ of co-located objects and the quotient has $100\%$ median LSCC. Cut rank $\rk_{\mathrm{SC}}(c)$ (Theorem~\ref{thm:cut-budget}) has median $48$ and p95 $123$ (LSCC reaches $100\%$ at $\rk=64$ median, $128$ p10).

\noindent\textbf{Exp-4: Recovery Budget $\rk$ and Query Work.}
With $L=16$ on \textsc{Overture-1M} ($\bar{\sigma}=10^{-1}$; \apx{app:extra-exp}), Recall@10 rises from $0.3730$ at $\rk=8$ ($16{,}864$~QPS) to $0.8576$ at $\rk=64$ ($3{,}243$~QPS) and $0.8691$ at $\rk=96$ ($2{,}257$~QPS, $1{,}330$ evaluations), nearly matching $\rk=\infty$ ($0.8738$, $4{,}157$ evaluations, $797$~QPS): from $\rk=8$ to $\infty$, vector evaluations grow by $37\times$ vs.\ $11\times$ for scans. Across \emph{all $20$ filtered cells}, the max-recall gap between $\rk=64$ and $\infty$ is $\le 0.0032$ ($\le 0.0008$ for $\bar{\sigma}\le 10^{-3}$), and the active fraction scored for $0.99$ recall drops from $1.0$ at $10^{-5}$ to $0.004$--$0.02$ at $\bar{\sigma}=1$.

\noindent\textbf{Exp-5--Exp-8: Ablations, Sensitivity, Scale, and Fibers.}
\emph{Exp-5 (Figure~\ref{fig:core-entry-ablation}(b)--(d)):} On \textsc{Yfcc-Clip} ($\bar{\sigma}=10^{-2}, \rk=32$), removing geographic ($B_g=0$) or vector edges ($B_v=0$) caps recall at $0.8811$ or $0.9285$ vs.\ $0.9981$ for both; grid seeding matches graph descent while reducing seed time by $18\times$ ($1.2$ vs.\ $21.7\,\mu\mathrm{s}$).
\emph{Exp-6 (Figure~\ref{fig:robustness}(a),(b),(d)):} $B=384$ matches $B=512$ at lower scan cost; $B_v/B_g=128/256$ balances scans ($2{,}611$ vs.\ $3{,}524$ for $192/192$ at $L=16$) and recall ($0.9590$ vs.\ $0.9513$ for $64/320$ at $L=64$); and $k\in[5,100]$ exceed $0.997$ recall.
\emph{Exp-7 ($10^6\to 10^7$):} From \textsc{Sift-Geo-1M} to \textsc{SpaceV-Geo-10M}, size scales by $10.0\times$ ($1.44\to 14.34$~GiB) and build time by $17.1\times$ ($117\to 2{,}007$~s, vs.\ $43.8\times$ for CubeGraph).
\emph{Exp-8 (Figure~\ref{fig:robustness}(c)):} On \textsc{Yfcc-Clip} (fibers up to $5{,}726$), fiber $\alpha$-pruning beats ID truncation ($0.9928$ vs.\ $0.9860$ recall at $L=256$; $0.9981$ at $L=1024$).

%% file: sections/06_related_work.tex
\section{Related Work}
\label{sec:related}

\noindent\textbf{Filtered and Hybrid ANN Search.}
Proximity graphs searched with beam search (HNSW~\cite{hnsw}, NSG~\cite{nsg}, DiskANN~\cite{diskann}) lead unfiltered ANN~\cite{DBLP:journals/pvldb/WangXY021,DBLP:journals/pacmmod/AziziEP25,DBLP:journals/is/AumullerBF20,Pan_2024}. For attribute filters~\cite{Chronis_2025}, Filtered-DiskANN~\cite{filteredDiskANN}, UNG~\cite{ung}, and JAG~\cite{jag} embed discrete labels into edges; CGIF~\cite{cgif} pairs HNSW with inverted files; SIEVE~\cite{sieve} builds workload-tuned sub-indexes; VBASE~\cite{vbase} uses relaxed monotonicity; NHQ~\cite{nhq} and DEG~\cite{deg} define composite attribute--vector metrics; and ACORN~\cite{acorn} filters two-hop neighborhoods. However, 2D disks cannot be indexed by finite inverted lists and act as hard boundaries ($\dgeo(s_v,g_q)\le r_q$). Range graphs (SeRF~\cite{serf}, uRNG~\cite{urng}, RNSG~\cite{rnsg}, WinFilter~\cite{window}, iRangeGraph~\cite{Xu_2024}, UNIFY~\cite{Liang_2024}, DSG~\cite{Peng_2025}, DIGRA~\cite{Jiang_2025}) rely on a 1D attribute order absent for 2D disks.

\noindent\textbf{Spatial--Vector and Fiber-Navigable Retrieval.}
Spatial~\cite{Guttman_1984,Samet_1984,Gaede_1998} and geo-textual indexes~\cite{chen2021survey} prune 2D space before secondary filtering (\emph{Spatial-First}, \emph{Geo-Partition}); \method uses a flat grid only for $\tau$ entry seeds. Mesh~\cite{DBLP:journals/vldb/SongYCYXLC25} tunes cell graphs for a known workload and switches between spatial scan and cell-graph search by selectivity, whereas CubeGraph~\cite{https://doi.org/10.48550/arxiv.2604.06616} builds graphs across tree levels. Unlike FNS~\cite{fiberann}, which calls any predicate-induced subgraph a ``fiber,'' our \emph{coordinate fibers} $F(z)$ are co-located object classes admitted in full or not at all by any disk (Lemma~\ref{lem:fiber-closure}), which connects intra-fiber $\alpha$-pruning with RGG theory~\cite{Penrose_1999,Penrose_2003,Balister_2005,PenroseYang_2025,DaleyVereJones_2003}.

%% file: sections/07_conclusion.tex
\section{Conclusion}
\label{sec:conclusion}

We proposed \method for \problem search via \emph{Ranked Active-Graph Recovery}: (1)~filtering a vector-sorted row preserves survivor order, recovering $\RecG_{\rk}(c)$ governed by cut rank $\rk_{\mathrm{SC}}(c)$; and (2)~coordinate fibers pass or fail any disk filter as a unit, yielding polylogarithmic-degree connectivity and $2.49\times$ higher selective QPS.

%% file: appendix/theory_proofs.tex
\section{Complete Proofs of Theoretical Results}
\label{app:proofs}

This appendix provides the complete formal proofs of all lemmas, propositions, theorems, and corollaries stated in Sections~\ref{sec:preliminaries}, \ref{sec:theory}, and~\ref{sec:index}.

\subsection{Online Realization Invariance and Omitted Main-Text Proofs}
\label{app:main-proofs}

\begin{proof}[Proof of Theorem~\ref{thm:cut-budget}]
Part~(a) holds because Definition~\ref{def:recovery-graph} keeps the vertex set $\Aset(c)$ while selecting at most $\min\{\rk,|C_u(c)|\}\le\min\{\rk,B\}$ outgoing edges per vertex. Part~(b) follows because relaxing $\operatorname{rank}_{u,c}(v)\le\rk$ only adds edges, and $|C_u(c)|\le B$. For Part~(c), a finite directed graph is strongly connected iff every nonempty proper subset $S\subsetneq\Aset(c)$ has at least one outgoing edge to $\Aset(c)\setminus S$. Under Definition~\ref{def:recovery-graph}, the lowest-rank active edge leaving $S$ in $\RecG_{\infty}(c)$ has rank $\min_{u\in S,v\in C_u(c)\setminus S}\operatorname{rank}_{u,c}(v)$, which is retained in $\RecG_{\rk}(c)$ across every cut $S$ iff $\rk\ge\rk_{\mathrm{SC}}(c)$.
\end{proof}

\begin{proof}[Proof of Lemma~\ref{lem:fiber-closure} and Theorem~\ref{thm:factorization}]
For Lemma~\ref{lem:fiber-closure}, all $v\in F(z)$ share $s_v=z$ and distance $\dgeo(z,g)$ to $g$, so $\dgeo(z,g)\le r$ holds either for all of $F(z)$ or for none. For Theorem~\ref{thm:factorization}, given $u\in F(z_0)$ and $v\in F(z_t)$ in $\Aset(c)$, Condition~(ii) supplies a quotient path $z_0\to\cdots\to z_t$ witnessed by edges $(a_i,b_{i+1})\in E$ with $a_i\in F(z_i),b_{i+1}\in F(z_{i+1})$, and Condition~(i) supplies intra-fiber paths $u\to a_0$ in $G[F(z_0)]$, $b_i\to a_i$ in $G[F(z_i)]$, and $b_t\to v$ in $G[F(z_t)]$.
\end{proof}

\begin{remark}[Visitation-Order Invariance of Algorithm~\ref{alg:expand}]
When Algorithm~\ref{alg:expand} expands an active vertex $u\in\Aset(c_q)$, Phase~1 scans $\rowu(u)$ in stored $\prec_u$ order, tests each entry against $\Aset(c_q)$, and increments the active ordinal counter $j$ on line~5 for \emph{every} active entry $v\in C_u(c_q)$ before checking whether $v\in\mathcal{V}$ on line~6. The scan halts on line~7 immediately after processing the entry with $\operatorname{rank}_{u,c_q}(v)=\rk$ (or at the end of the row when $|C_u(c_q)|<\rk$). Consequently, the set of active neighbors inspected at $u$ is $\{v\in C_u(c_q):\operatorname{rank}_{u,c_q}(v)\le\rk\}$, which coincides with the out-neighborhood of $u$ in $\RecG_{\rk}(c_q)$ and depends solely on $\rowu(u)$, $c_q$, and $\rk$, independent of the runtime visited set $\mathcal{V}$.
\end{remark}

\begin{proof}[Proof of Theorem~\ref{thm:end-to-end}]
The second inequality in~\eqref{eq:quota-scaling} satisfies the premise of Theorem~\ref{thm:fiber-connectivity}, which establishes Condition~(i) of Theorem~\ref{thm:factorization} for all active fibers with failure probability at most $\delta/3$. The first inequality in~\eqref{eq:quota-scaling} satisfies the premise of Theorem~\ref{thm:quotient-connectivity}, which establishes Condition~(ii) of Theorem~\ref{thm:factorization} with failure probability at most $2\delta/3$. By the union bound, both conditions hold simultaneously with probability at least $1-\delta$, and Theorem~\ref{thm:factorization} implies that $G[\Aset(c)]$ is strongly connected. Because all conditional mark bounds hold uniformly across realizations of the site locations and contexts, taking expectations over the site and context distributions preserves the $1-\delta$ bound under arbitrary location--vector dependence permitted by Assumption~\ref{asm:marked-sites}.
\end{proof}

\begin{proof}[Proof of Corollary~\ref{cor:navigability}]
By Part~(b) of Theorem~\ref{thm:quotient-connectivity}, on the event $L_u(\rho_c)\le a_\alpha(h_c+1)\le B_g$, the geographic stream of any active vertex $u\in\Aset(c)$ visits every active coordinate within great-circle distance $\rho_c$ without exhausting its geographic quota $B_g$. Under the compact-disk condition $\operatorname{diam}(\Omega_c)\le\rho_c$, every active target $v\in\Aset(c)$ satisfies $\dgeo(s_u,s_v)\le\operatorname{diam}(\Omega_c)\le\rho_c$; therefore, $u$ visits $F(s_v)$ and, because $K_f\ge |F(s_v)|-1$ and $B_g$ does not run out, examines $v$ under the intra-fiber $\alpha$-rule. Thus, either $u$ retains the direct edge $(u,v)$ (giving $w=v$ with $0 \le \dvec(x_u,x_v)/\alpha$), or $v$ is pruned in favor of an already retained witness $w\in F(s_v)$ satisfying $\dvec(x_w,x_v)\le\dvec(x_u,x_v)/\alpha$. Once inside $F(s_v)$, subsequent hops remain in $F(s_v)$ and contract the remaining distance to $v$ by at least $\alpha$ per hop by Lemma~\ref{lem:fiber-witness} (with $r=D_x$). Because distinct objects are separated by at least $\varepsilon_\alpha$ on the charge event, greedy routing reaches $v$ in at most $\lceil\log_\alpha(D_x/\varepsilon_\alpha)\rceil$ hops. Finally, for an arbitrary query vector $x_q$, let $u\in\Aset(c)$ be a local minimum of greedy search (so $\dvec(x_q,x_u)\le\dvec(x_q,x_w)$ for all active out-neighbors $w$ of $u$) and let $v^\star = \arg\min_{v\in\Aset(c)}\dvec(x_q,x_v)$. Writing $A=\dvec(x_q,x_u)$ and $B^\star=\dvec(x_q,x_{v^\star})$, the witness $w$ for $v^\star$ satisfies $\dvec(x_w,x_{v^\star})\le\dvec(x_u,x_{v^\star})/\alpha \le (A+B^\star)/\alpha$. By the triangle inequality, $A \le \dvec(x_q,x_w) \le B^\star + \dvec(x_w,x_{v^\star}) \le B^\star + (A+B^\star)/\alpha$, which rearranges to $A \le B^\star(\alpha+1)/(\alpha-1)$.
\end{proof}

\begin{corollary}[Simultaneous Guarantee Across a Query Workload]
\label{cor:workload}
Let $\mathcal{C}$ be any workload of $Q$ query disks satisfying $M\pi_c \ge \log(6Q/\delta)$ for all $c\in\mathcal{C}$. Replacing $\delta$ by $\delta/Q$ in Equations~\eqref{eq:alpha-charge}--\eqref{eq:quota-scaling} ensures that the single constructed graph $G$ satisfies Theorem~\ref{thm:end-to-end} for all $Q$ disks simultaneously with probability at least $1-\delta$.
\end{corollary}

\begin{proof}[Proof of Corollary~\ref{cor:workload}]
Apply Theorem~\ref{thm:end-to-end} with target failure probability $\delta' = \delta/Q$ to each of the $Q$ disks in $\mathcal{C}$, and take a union bound over the $Q$ events. Because $\delta'$ enters the degree bounds~\eqref{eq:alpha-charge}--\eqref{eq:quota-scaling} only inside logarithms, the required quotas grow only additively by $O(\log Q)$ while using a single static graph.
\end{proof}

\begin{proof}[Proof of Proposition~\ref{prop:assembly-contract}]
For Item~(i), every candidate $v$ accepted at line~10 is placed in $P$ (if not already present), and every candidate rejected at line~9 is witnessed by an earlier member $w\in W\subseteq P\cap F(z)$. For Item~(ii), the counter $g$ increments by exactly $\chi_u(z)$ for each visited coordinate $z$; thus $L_u(\rho)\le B_g$ ensures that the outer loop reaches every coordinate within distance $\rho$ before $g$ reaches $B_g$, and the first candidate examined in each nonempty fiber $F(z)$ ($z\ne s_u$) is always accepted into $W\subseteq P$ because $W$ is initially empty. Item~(iii) follows directly from deduplication, $B_v+B_g\le B$, and the final sort at line~14.
\end{proof}

\subsection{Exact Active-Neighbor Special Case}
\label{app:realization}

Although strong connectivity and navigability do not require $\RecG_{\rk}(c)$ to match an exact $j$-nearest-neighbor graph edge-for-edge, our recovery operator reproduces the exact active $j$-NN neighborhood whenever the row contains those neighbors. Specifically, let $N_j^\star(u,c)$ denote the exact top-$j$ active vector neighbors of $u$ in $\Aset(c)\setminus\{u\}$ ordered by $\prec_u$. Then
\begin{equation}
\label{eq:exact-special-case}
  N_j^\star(u,c)\subseteq\rowu(u)
  \quad\Longleftrightarrow\quad
  N^+_{\RecG_j(c)}(u) = N_j^\star(u,c).
\end{equation}
To verify~\eqref{eq:exact-special-case}, observe that both $\rowu(u)$ and $N_j^\star(u,c)$ are ordered by the same strict total order $v \prec_u w \iff (\dvec(x_u,x_v),v) <_{\mathrm{lex}} (\dvec(x_u,x_w),w)$. If $N_j^\star(u,c)\subseteq\rowu(u)$, then removing inactive entries from $\rowu(u)$ leaves $C_u(c)$ as a subsequence of $\prec_u$ that contains the globally smallest $\min\{j,m-1\}$ elements of $\Aset(c)\setminus\{u\}$; hence its prefix of length $\min\{j,m-1\}$ is precisely $N_j^\star(u,c)$. The converse holds because every edge of $\RecG_j(c)$ at $u$ is drawn from $\rowu(u)$.

\subsection{Proof of Disk Range Complexity}
\label{app:range-complexity}

\begin{proof}[Proof of Theorem~\ref{thm:range-complexity}]
Map each point $s\in\mathbb{S}^2$ to its unit Cartesian vector $\mathbf{u}_s\in\mathbb{R}^3$. On a sphere of radius $R_{\oplus}$, a closed spherical disk $c=(g,r)$ with central angle $\theta = \min\{r/R_{\oplus},\pi\}\in[0,\pi]$ is given by
\[
  \{s\in\mathbb{S}^2 : \langle\mathbf{u}_s, \mathbf{u}_g\rangle \ge \cos\theta\},
\]
which is the intersection of $\mathbb{S}^2$ with a closed affine halfspace in $\mathbb{R}^3$. Because the class of affine halfspaces in $\mathbb{R}^3$ has VC dimension $4$~\cite{Vapnik_1971,Matousek_2002}, restricting the domain to $\mathbb{S}^2\subset\mathbb{R}^3$ cannot increase the VC dimension. By the Sauer--Shelah lemma~\cite{Sauer_1972}, the number of distinct subsets induced on $n$ points is at most $\sum_{i=0}^{4}\binom{n}{i} = O(n^4)$.

For the $\Omega(n^2)$ lower bound, place $n$ points $s_1,\ldots,s_n$ along a short geodesic arc of $\mathbb{S}^2$ at equal great-circle spacing $\eta < \pi R_{\oplus}/(2n)$. For any index pair $1\le i\le j\le n$, let $g_{i,j}$ be the geodesic midpoint between $s_i$ and $s_j$ along the arc, and set the disk radius to $r_{i,j} = (j-i)\eta/2$. Then $\Aset(g_{i,j},r_{i,j}) = \{s_i, s_{i+1}, \ldots, s_j\}$. These $\binom{n}{2}+n$ nonempty contiguous intervals (together with $\varnothing$ at $r=0$ away from the arc) are all distinct, yielding $n(n+1)/2 + 1 = \Omega(n^2)$ active sets.
\end{proof}

\subsection{Proof of Intra-Fiber Connectivity (Theorem~\ref{thm:fiber-connectivity})}
\label{app:fiber-connectivity}

We first state the formal regularity conditions of the hierarchical marked-site model introduced in Section~\ref{sec:obligations} and prove the deterministic witness-lifting lemma.

\begin{assumption}[Regular Binomial Marked Sites]
\label{asm:marked-sites}
\emph{(a)~Spatial regularity:} $0<p_{s,-}\le p_s(z)\le p_{s,+}<\infty$, and every connected disk region $\Omega_c=\mathcal{S}\cap\mathcal{B}(g,r)$ satisfies $\lambda_2(\Omega_c\cap\mathcal{B}(z,t))\ge c_s t^2$ for all $z\in\Omega_c$ and $0<t\le\operatorname{diam}(\Omega_c)$.
\emph{(b)~Vector regularity:} $(\mathcal{X},\dvec)$ has diameter $D_x$ and doubling constant $C_D$. Each support $\mathcal{X}_{z,\theta}$ is compact, connected, and geodesic of intrinsic dimension $d_{\mathcal{X}}>0$ with atomless distances, and for constants $r_x,q_{x,-},q_{x,+},C_x>0$ and all $0<r\le r_x$, $x\in\mathcal{X}_{z,\theta}$, $y\in\mathcal{X}$: $q_{x,-}r^{d_{\mathcal{X}}}\le\nu_{z,\theta}(\mathcal{B}(x,r))$, $\nu_{z,\theta}(\mathcal{B}(y,r))\le q_{x,+}r^{d_{\mathcal{X}}}$, and covering number $\mathcal{N}(\mathcal{X}_{z,\theta},r)\le C_x r^{-d_{\mathcal{X}}}$, with minimum separation scale $\varepsilon_\alpha=\min\{r_x,D_x/2,(\delta/(3n^2 q_{x,+}))^{1/d_{\mathcal{X}}}\}$.
\end{assumption}

\begin{lemma}[Intra-Fiber Witness Lifting]
\label{lem:fiber-witness}
If every vertex $u\in F(z)$ has at most $K_f$ other fiber members within vector distance $r$, then every $r$-radius pair $(u,v)$ in $F(z)$ is connected in $G[F(z)]$ by a directed path of at most $\lceil\log_\alpha(\dvec(x_u,x_v)/\varepsilon_z)\rceil$ hops (where $\varepsilon_z=\min_{u'\ne v'\in F(z)}\dvec(x_{u'},x_{v'})>0$) that contracts vector distance to $x_v$ by at least $\alpha$ per hop.
\end{lemma}

\begin{proof}[Proof of Lemma~\ref{lem:fiber-witness}]
Fix any $u,v\in F(z)$ with $0 < \dvec(x_u,x_v) \le r$. During construction for owner $u$, $F(z)=F(s_u)$ is the first fiber visited by the geographic stream. By assumption, at most $K_f \le \min\{B_g, C_{\mathrm{tie}}\}$ other members of $F(z)$ have vector distance at most $r$ from $x_u$. Because candidates in $F(z)\setminus\{u\}$ are examined in ascending order of $\dvec(x_u,\cdot)$, candidate $v$ is guaranteed to be examined before either the scan cap $C_{\mathrm{tie}}$ or the quota $B_g$ is reached.

When $v$ is examined, if $v$ is accepted, then $(u,v)\in E(G[F(z)])$ forms a 1-hop path with distance $0 \le \dvec(x_u,x_v)/\alpha$. Otherwise, $v$ is rejected at line~12 of Algorithm~\ref{alg:row-assembly} because some already accepted survivor $w_1\in F(z)$ satisfies $\dvec(x_{w_1},x_v) \le \dvec(x_u,x_v)/\alpha < \dvec(x_u,x_v) \le r$. Thus $(w_1,v)$ is itself a strictly shorter $r$-radius pair inside $F(z)$. Repeating this argument inductively from $w_1$ generates a sequence $u=w_0 \to w_1 \to \cdots \to w_t = v$ in $G[F(z)]$ along which $\dvec(x_{w_j},x_v) \le \alpha^{-j}\dvec(x_u,x_v)$. Since any two distinct vertices in $F(z)$ have distance at least $\varepsilon_z$, the penultimate vertex satisfies $\varepsilon_z \le \dvec(x_{w_{t-1}},x_v) \le \alpha^{-(t-1)}\dvec(x_u,x_v)$, which implies $t \le \lceil\log_\alpha(\dvec(x_u,x_v)/\varepsilon_z)\rceil$. Finally, if the symmetric $r$-radius graph on $F(z)$ is connected, replacing both orientations of every $r$-radius edge by its lifted directed path makes $G[F(z)]$ strongly connected.
\end{proof}

\begin{proof}[Proof of Theorem~\ref{thm:fiber-connectivity}(a)]
Fix a fiber $F(z)$ of multiplicity $f = |F(z)|$. If $f - 1 \le K_f$, then every vertex $u\in F(z)$ examines all $f-1$ other members of $F(z)$, so Lemma~\ref{lem:fiber-witness} applies with $r = D_x$ deterministically (with zero failure probability). Suppose now that $f - 1 > K_f$. Conditioned on the site $Z_i=z$ and context $\Theta_i=\theta$, choose a scale $r_f > 0$ with $3r_f \le r_x$ and define $\mu_f = (f-1)q_{x,+}(3r_f)^{d_{\mathcal{X}}}$. We claim that $G[F(z)]$ satisfies Parts~(a) and~(b) except with probability at most
\begin{equation}
\label{eq:delta-fiber}
  \delta_{\mathrm{fib}}(z) =
  C_x(r_f/2)^{-d_{\mathcal{X}}} \exp\!\bigl(-f q_{x,-}(r_f/2)^{d_{\mathcal{X}}}\bigr)
  + f\left(\frac{e\mu_f}{K_f+1}\right)^{K_f+1},
\end{equation}
whenever $K_f + 1 \ge \mu_f$. To see this, cover the geodesic support $\mathcal{X}_{z,\theta}$ by at most $C_x(r_f/2)^{-d_{\mathcal{X}}}$ balls of radius $r_f/2$. By Assumption~\ref{asm:marked-sites}, each ball has $\nu_{z,\theta}$-probability at least $q_{x,-}(r_f/2)^{d_{\mathcal{X}}}$, so the probability that any cover ball receives zero marks among the $f$ i.i.d.\ draws is bounded by the first term of~\eqref{eq:delta-fiber}. When every cover ball is nonempty, the fiber marks form an $r_f$-net of the connected geodesic space $\mathcal{X}_{z,\theta}$, which implies that the symmetric $3r_f$-radius graph on $F(z)$ is connected. Next, conditioned on $x_u$, the number of other fiber marks falling inside $\mathcal{B}(x_u, 3r_f)$ is stochastically dominated by $\operatorname{Bin}(f-1, q_{x,+}(3r_f)^{d_{\mathcal{X}}})$ with mean $\mu_f$. By the multiplicative Chernoff bound for $K_f+1 \ge \mu_f$ and a union bound over the $f$ vertices of $F(z)$, the probability that any vertex $u\in F(z)$ has more than $K_f$ fiber neighbors within radius $3r_f$ is at most the second term of~\eqref{eq:delta-fiber}. On the intersection of both events, Lemma~\ref{lem:fiber-witness} (applied at radius $r = 3r_f$) establishes both (a) and (b).

To obtain the uniform logarithmic bound, set $L_0 = \log(3nM/\delta)$ and choose $r_f^{d_{\mathcal{X}}} = A L_0 / f$, where $A = 2^{d_{\mathcal{X}}+1}/q_{x,-}$. Because $f > K_f \ge C_f L_0$, choosing
\[
  C_f \ge \max\!\left\{A(3/r_x)^{d_{\mathcal{X}}},\; \frac{C_x q_{x,-}}{3},\; e^2 q_{x,+} 3^{d_{\mathcal{X}}} A,\; 2\right\}
\]
ensures first that $3r_f < 3(A/C_f)^{1/d_{\mathcal{X}}} \le r_x$. Second, substituting $r_f^{d_{\mathcal{X}}} = A L_0 / f$ into~\eqref{eq:delta-fiber} bounds the first term by $\frac{C_x q_{x,-} f}{2 L_0}(\frac{\delta}{3nM})^2 \le \frac{\delta}{6M}$ (since $f\le n$ and $f > C_f L_0 \ge C_x q_{x,-}/3$). Third, because $K_f + 1 > C_f L_0 \ge e^2 \mu_f$, the second term of~\eqref{eq:delta-fiber} is at most $f e^{-(K_f+1)} \le n(\frac{\delta}{3nM})^2 \le \frac{\delta}{6M}$. Thus $\delta_{\mathrm{fib}}(z) \le \delta/(3M)$ for each fiber, and a union bound over all $M$ fibers completes the proof.
\end{proof}

\paragraph{Derivation of the Per-Coordinate $\alpha$-Charge Bound~\eqref{eq:alpha-charge}.}
Let $W=(w_1,\ldots,w_t)$ be the survivors accepted by owner $u$ from a single coordinate fiber $F(z)$ in ascending order of $\dvec(x_u,\cdot)$. For any $1\le i < j\le t$, the acceptance of $w_j$ implies $\dvec(x_{w_i},x_{w_j}) > \dvec(x_u,x_{w_j})/\alpha$. Partition $\mathcal{X}\setminus\{x_u\}$ into dyadic annuli $\mathcal{A}_k = \{x : 2^{-k}D_x < \dvec(x_u,x) \le 2^{-(k-1)}D_x\}$. Any two survivors $w_i,w_j\in\mathcal{A}_k$ lie inside $\mathcal{B}(x_u,R)$ with $R = 2^{-(k-1)}D_x$ and are separated by more than $R/(2\alpha)$. Iterating the doubling property of $\mathcal{X}$ $m=\lceil\log_2(4\alpha)\rceil$ times covers $\mathcal{B}(x_u,R)$ with at most $P_\alpha = C_D^m$ balls of radius $2^{-m}R \le R/(4\alpha)$ (diameter at most $R/(2\alpha)$), so each ball contains at most one survivor and $|\mathcal{A}_k\cap W|\le P_\alpha$. Whenever $\min_{v\in F(z)\setminus\{u\}}\dvec(x_u,x_v) \ge \varepsilon_\alpha$, only annuli with $2^{-(k-1)}D_x \ge \varepsilon_\alpha$ can be nonempty (at most $1+\lceil\log_2(D_x/\varepsilon_\alpha)\rceil$ annuli), capping the total number of survivors from $F(z)$ at $a_\alpha$ in~\eqref{eq:alpha-charge}. Finally, conditioned on $X_u=x_u$ and the site contexts, $\Pr\{\min_{v\in F(z)\setminus\{u\}}\dvec(x_u,X_v) < \varepsilon_\alpha\} \le f_z q_{x,+}\varepsilon_\alpha^{d_{\mathcal{X}}}$; summing over all $u\in V$ and $z$ (using $\sum_z f_z = n$) shows that the minimum separation $\varepsilon_\alpha$ in~\eqref{eq:alpha-charge} holds globally with probability at least $1-\delta/3$.

\subsection{Proof of Coordinate-Quotient Connectivity (Theorem~\ref{thm:quotient-connectivity})}
\label{app:site-backbone}

\begin{proof}[Proof of Theorem~\ref{thm:quotient-connectivity}(b)]
We first show that the symmetric $\rho_c$-radius graph over the active sites $\{Z_i\in\Omega_c\}$ is connected (and nonempty) except with probability at most $\delta/6$:
\begin{itemize}[leftmargin=1.5em,nosep]
  \item \emph{Case 1 ($\rho_c \ge \operatorname{diam}(\Omega_c)$):} Every pair of active sites in $\Omega_c$ lies at distance at most $\operatorname{diam}(\Omega_c)\le\rho_c$, so the $\rho_c$-radius graph on $\{Z_i\in\Omega_c\}$ is a complete graph. Moreover, because $M\pi_c \ge \log(6/\delta)$, the probability that no site falls in $\Omega_c$ is $(1-\pi_c)^M \le e^{-M\pi_c} \le \delta/6$.
  \item \emph{Case 2 ($\rho_c < \operatorname{diam}(\Omega_c)$):} By the interior-area regularity in Assumption~\ref{asm:marked-sites} (applied at scale $t \le \rho_c/3 < \operatorname{diam}(\Omega_c)$), we can tile the connected region $\Omega_c$ into at most $N_{\mathrm{cell}} \le C_0 \lambda_2(\Omega_c)/\rho_c^2$ connected cells of diameter at most $\rho_c/3$ such that any two points in adjacent cells are within great-circle distance $\rho_c$, and each cell $\kappa$ has surface area at least $c_0 \rho_c^2$ (hence site mass at least $c_0 p_{s,-}\rho_c^2$). Under $M$ independent draws from $p_s$, a given cell $\kappa$ is empty with probability at most $(1 - c_0 p_{s,-}\rho_c^2)^M \le \exp(-M c_0 p_{s,-}\rho_c^2)$. Because $\pi_c \le p_{s,+}\lambda_2(\Omega_c)$, substituting $\rho_c$ from~\eqref{eq:alpha-charge} with $C_s^2 \ge \max\{\frac{2 p_{s,+}}{c_0 p_{s,-}},\,\frac{2C_0}{3}\}$ and taking a union bound over all $N_{\mathrm{cell}}$ cells ensures that every cell contains at least one site with failure probability at most $N_{\mathrm{cell}}(\delta/(3M))^2 \le \delta/6$. When every cell is occupied, the symmetric $\rho_c$-radius graph over $\{Z_i\in\Omega_c\}$ is connected~\cite{Penrose_1999,Penrose_2003}.
\end{itemize}

Next, for any fixed site $z\in\Omega_c$, the spherical ball $\mathcal{B}(z,\rho_c)\cap\mathcal{S}$ has surface area at most $c_1 \rho_c^2$, and hence $p_s$-mass at most $c_1 p_{s,+}\rho_c^2$. Since $\pi_c \ge p_{s,-}\lambda_2(\Omega_c)$, the expected number of other sites within distance $\rho_c$ of $z$ is at most $\mu_s = (M-1)c_1 p_{s,+}\rho_c^2 \le \frac{c_1 p_{s,+} C_s^2}{p_{s,-}}\log(3M/\delta)$. Applying the multiplicative Chernoff bound and a union bound over all $M$ sites shows that, for $C_h \ge \max\{e^2 c_1 (p_{s,+}/p_{s,-}) C_s^2,\, 2\}$, every site has at most $h_c = \lceil C_h\log(3M/\delta)\rceil$ other sites within distance $\rho_c$ except with probability at most $M e^{-(h_c+1)} \le \delta/6$.

Combining the $\alpha$-charge event (failure probability $\le \delta/3$), the cell-occupancy event ($\le \delta/6$), and the site-crowding event ($\le \delta/6$) gives an intersection event of probability at least $1 - 2\delta/3$. On this intersection, every active vertex $u\in\Aset(c)$ encounters at most $h_c$ other coordinates within distance $\rho_c$ and retains at most $a_\alpha$ edges from each of them (and from its own fiber), so $L_u(\rho_c) \le a_\alpha(h_c+1) \le B_g$. Therefore, $u$'s nearest-first geographic stream reaches every active coordinate within distance $\rho_c$ and retains at least one representative from each before exhausting $B_g$, proving Part~(b). Because this holds in both directions for every $\rho_c$-close pair of active coordinates, the coordinate quotient contains the connected symmetric $\rho_c$-radius site graph and is therefore strongly connected, proving Part~(a).
\end{proof}

\subsection{Explicit Finite-Sample Connectivity Bound}
\label{app:finite-bound}

Theorem~\ref{thm:end-to-end} sets the annulus scale $\varepsilon_\alpha$, the site linking radius $\rho_c$, and the crowding bound $h_c$ to their asymptotic logarithmic forms. Leaving these three scales as free parameters yields a non-asymptotic finite-sample failure bound in terms of the problem parameters.

Let $C_0, c_0, c_1 > 0$ be the spherical domain constants defined in Appendix~\ref{app:site-backbone}. For any $0 < \varepsilon \le \min\{r_x, D_x/2\}$, $0 < \rho \le \operatorname{diam}(\Omega_c)$, and integer $h_\rho \ge 0$, define
\[
\begin{aligned}
  a(\varepsilon) &= \min\!\left\{K_f,\; P_\alpha\left(1+\left\lceil\log_2(D_x/\varepsilon)\right\rceil\right)\right\},\\
  \mu_s(\rho) &= (M-1)c_1 p_{s,+}\rho^2,
\end{aligned}
\]
and the four-term failure sum
\begin{align}
\label{eq:finite-failure}
  \Pi(c) ={}& n^2 q_{x,+}\varepsilon^{d_{\mathcal{X}}}
  + C_0\frac{\lambda_2(\Omega_c)}{\rho^2}\exp\!\bigl(-M c_0 p_{s,-}\rho^2\bigr) \notag\\
  &+ M\left(\frac{e\mu_s(\rho)}{h_\rho+1}\right)^{h_\rho+1}
  + \sum_{i=1}^{M}\delta_{\mathrm{fib},i},
\end{align}
where $\delta_{\mathrm{fib},i}$ is given by~\eqref{eq:delta-fiber} at scale $r_i$ for non-singleton fibers and $\delta_{\mathrm{fib},i}=0$ for singleton fibers ($f_i=1$). The four terms in~\eqref{eq:finite-failure} correspond one-to-one to the four failure modes: (1)~violation of the per-coordinate charge bound $a(\varepsilon)$, (2)~an empty cell in the geographic $\rho$-tiling of $\Omega_c$, (3)~more than $h_\rho$ sites within radius $\rho$ of some site, and (4)~loss of internal connectivity in some multi-object fiber.

\begin{theorem}[Explicit Finite-Sample Connectivity Bound]
\label{thm:finite-bound}
Conditioned on any multiplicity vector $(f_1,\ldots,f_M)$, suppose that $0 < \rho \le \operatorname{diam}(\Omega_c)$, $3r_i \le r_x$ for all $i$, $h_\rho + 1 \ge \mu_s(\rho)$, and the quotas satisfy
\[
  B_g \ge a(\varepsilon)(h_\rho + 1),
  \qquad
  K_f + 1 \ge \max_{1\le i\le M}\bigl\{(f_i-1)q_{x,+}(3r_i)^{d_{\mathcal{X}}}\bigr\}.
\]
Then the complete active projection is strongly connected with probability at least
\[
  \Pr\bigl\{G[\Aset(c)] \text{ is strongly connected}\bigr\} \ge 1 - \Pi(c).
\]
\end{theorem}

\begin{proof}
By the derivations in Appendices~\ref{app:fiber-connectivity} and~\ref{app:site-backbone}, the four bad events fail with probabilities bounded by the four respective summands of $\Pi(c)$ in~\eqref{eq:finite-failure}. On the complement of their union, every coordinate charges at most $a(\varepsilon)$ edges, every site has at most $h_\rho$ neighbors within distance $\rho$, the $\rho$-radius graph on active sites is connected, and every active fiber is internally strongly connected. Because $B_g \ge a(\varepsilon)(h_\rho+1)$, both conditions of Theorem~\ref{thm:factorization} hold, and the union bound yields $1-\Pi(c)$. Substituting $\varepsilon=\varepsilon_\alpha$, $\rho=\rho_c$, and $h_\rho=h_c$ recovers Theorem~\ref{thm:end-to-end}.
\end{proof}

\subsection{Route Diameter on Broader Disks}
\label{app:route-length}

When a query disk is broader than a single geographic budget radius ($\operatorname{diam}(\Omega_c) > \rho_c$), the active projection $G[\Aset(c)]$ still possesses a bounded directed diameter: a path can first walk across the geographic site backbone in $O(\dgeo(s_u,s_v)/\rho_c)$ hops to reach the target's fiber $F(s_v)$, and then contract in vector space inside $F(s_v)$ in at most $1+\lceil\log_\alpha(D_x/\varepsilon_\alpha)\rceil$ hops.

\begin{proposition}[Two-Phase Route Length Bound]
\label{prop:descent}
\label{cor:route-length}
On the intersection event of Theorem~\ref{thm:end-to-end} (with $K_f\ge |F(s_v)|-1$), for every pair of active vertices $u,v\in\Aset(c)$ separated by geodesic distance $D$ within $\Omega_c$, the active projection $G[\Aset(c)]$ contains a directed path from $u$ to $v$ of total hop length at most
\[
  O\!\left(\frac{D}{\rho_c}\right) + 1 + \left\lceil\log_\alpha\!\left(\frac{D_x}{\varepsilon_\alpha}\right)\right\rceil.
\]
\end{proposition}

\begin{proof}
Because every cell of diameter at most $\rho_c/3$ in the tiling of $\Omega_c$ contains at least one active site and every active vertex retains a representative of every active coordinate within distance $\rho_c$ (Theorem~\ref{thm:quotient-connectivity}(b)), walking along the cells covering a geodesic from $s_u$ to $s_v$ in $\Omega_c$ produces a chain of $O(D/\rho_c)$ inter-coordinate hops from $u$ to an active vertex $u'\in\Aset(c)$ with $\dgeo(s_{u'}, s_v) \le \rho_c$. Because $\dgeo(s_{u'}, s_v) \le \rho_c$, the local witness argument in the proof of Corollary~\ref{cor:navigability} supplies one hop from $u'$ to a witness $w\in F(s_v)$ followed by at most $\lceil\log_\alpha(D_x/\varepsilon_\alpha)\rceil$ contracting hops inside $F(s_v)$ to reach $v$.
\end{proof}

%% file: appendix/additional_details.tex
\section{Additional System and Experimental Details}
\label{app:additional}

\subsection{Summary of Notation and System Overview}
\label{app:notation}
\label{app:overview}

Table~\ref{tab:notation} summarizes the notation used throughout the paper, and Figure~\ref{fig:overview} illustrates the six-module offline and online workflow of \method. Offline construction produces two persistent data structures: the $4nB$-byte table of vector-sorted rows and the $O(n)$ spatial grid seed directory. Online, the query disk $c_q$ defines the activity predicate evaluated during sequential row scanning, while the grid directory supplies $\tau$ initial active seeds inside $\Aset(c_q)$.

\begin{table}[t]
  \centering
  \caption{Summary of frequently used notation.}
  \label{tab:notation}
  \footnotesize
  \setlength{\tabcolsep}{2pt}
  \begin{tabular}{@{}lp{0.67\columnwidth}@{}}
    \toprule
    \textbf{Symbol} & \textbf{Description}\\
    \midrule
    \multicolumn{2}{l}{\emph{Data and query model} (Section~\ref{sec:preliminaries})}\\
    $\dataset$, $V$, $n$ & object collection, vertex set $\{1,\ldots,n\}$, and cardinality\\
    $\mathcal{X}$, $\mathcal{S}$ & $D$-dimensional vector space and spherical geographic domain\\
    $x_v$, $s_v$ & feature vector and geographic position of object $v$\\
    $\dvec(\cdot,\cdot)$, $\dgeo(\cdot,\cdot)$ & Euclidean vector distance and spherical great-circle distance\\
    $c=(g,r)$, $\Omega_c$ & disk condition (center $g$, radius $r$) and its spherical region\\
    $\Aset(c)$, $m$, $\sigma(c)$ & active object set, active size $|\Aset(c)|$, and selectivity $m/n$\\
    $q=(x_q,c_q,k)$, $T_k(q)$ & \problem query and its exact top-$k$ active ground-truth set\\
    \midrule
    \multicolumn{2}{l}{\emph{Recovery graph and two-scale theory} (Section~\ref{sec:theory})}\\
    $G=(V,E)$, $\rowu(u)$ & directed physical graph and vector-sorted neighbor row of $u$\\
    $C_u(c)$ & active entries in $\rowu(u)$ under disk $c$, i.e., $\rowu(u)\cap\Aset(c)$\\
    $\RecG_\infty(c)$ & complete active projection $G[\Aset(c)]$ of the stored rows\\
    $\RecG_{\rk}(c)$, $\rk$ & recovery graph (first $\rk$ active row entries) and recovery budget\\
    $\rk_{\mathrm{SC}}(c)$ & exact cut-rank budget threshold for strong connectivity\\
    $F(z)$, $M$, $\delta$ & coordinate fiber $\{v:s_v=z\}$, distinct site count, failure budget\\
    $K_f$, $d_{\mathcal{X}}$ & per-fiber candidate cap ($\min\{B_g,C_{\mathrm{tie}}\}$) and intrinsic dimension\\
    $\chi_u(z)$, $L_u(\rho)$ & new geographic edges from $F(z)$ and cumulative quota used within $\rho$\\
    $\varepsilon_\alpha$, $a_\alpha$, $h_c$, $\rho_c$ & vector separation scale, per-fiber quota cap, crowding bound, and linking radius\\
    \midrule
    \multicolumn{2}{l}{\emph{\method index construction and search} (Section~\ref{sec:index})}\\
    $B$, $B_v$, $B_g$ & fixed row width ($B=B_v+B_g$), vector quota, and geographic quota\\
    $K_{\mathrm{vec}}$, $\alpha$, $C_{\mathrm{tie}}$ & vector pool size, diversity parameter ($\alpha>1$), and fiber scan cap\\
    $L$, $\tau$ & beam search list size ($\mathit{ef}_{\mathrm{search}}$) and spatial entry seed count\\
    \bottomrule
  \end{tabular}
\end{table}

\begin{figure*}[t]
  \centering
  \includegraphics[width=0.92\textwidth]{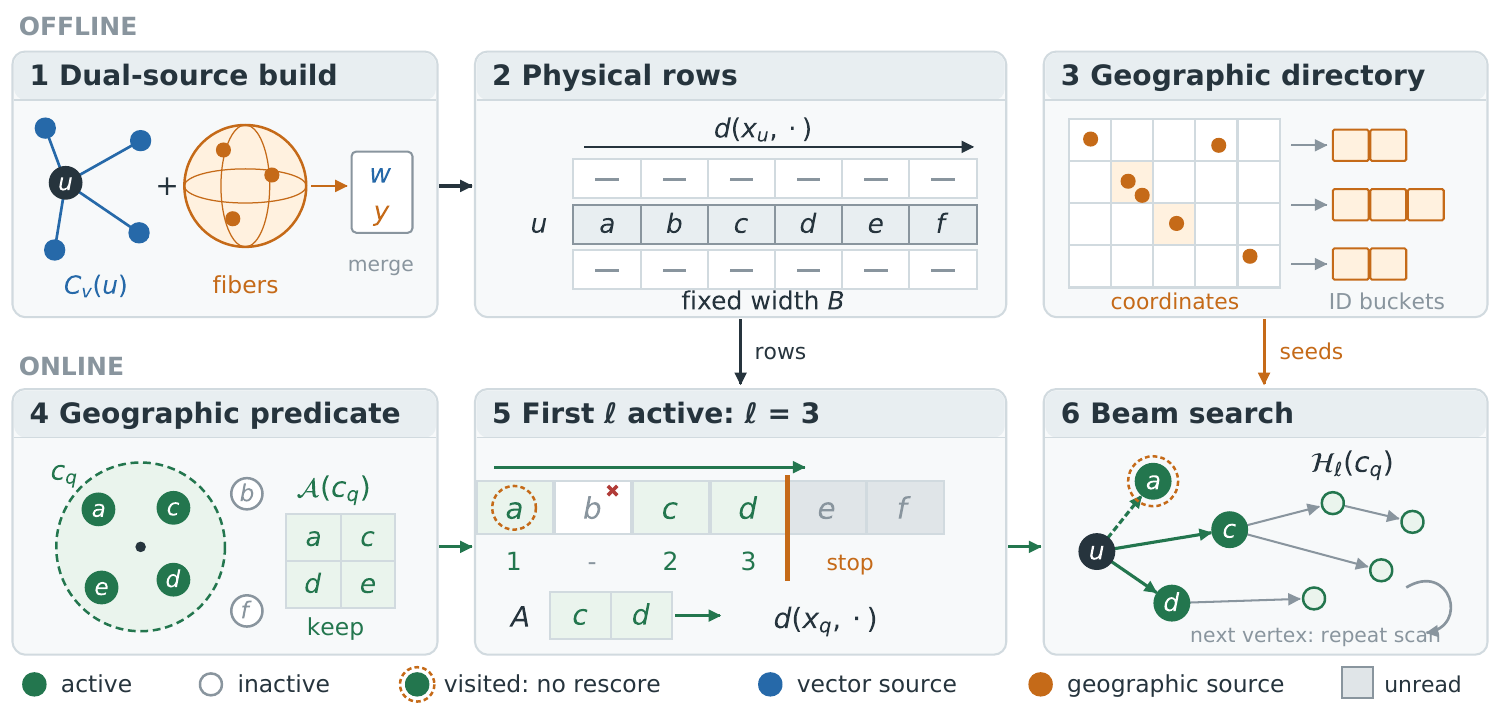}
  \caption{Six-module overview of \method. Offline, dual-source assembly produces fixed-width vector-sorted rows, while geographic coordinates populate a separate cell directory. Online, the disk predicate and stored row feed first-$\rk$-active scanning; the directory supplies entry seeds directly to beam search. The shared $\rk=3$ example counts visited $a$, skips inactive $b$, and stops at $d$, sending only $A=[c,d]$ to vector scoring. The beam traverses $\RecG_3(c_q)$, including edge $u\to a$, and repeats the scan at its next expanded vertex.}
  \Description{Six bounded modules use short local connectors. The upper modules show a vector graph plus spherical fibers, a fixed-width row matrix, and a geographic cell directory. The lower modules show a disk predicate, a compact first-three-active scan, and a recovery subgraph.}
  \label{fig:overview}
\end{figure*}

\subsection{Grid Directory and Entry Seeding}
\label{app:cell-index}
\label{app:entry-detail}

\begin{algorithm}[t]
\caption{\textsc{Seed}$(c_q)$: Active entry seed lookup from the grid directory}
\label{alg:seed}
\KwIn{disk condition $c_q=(g_q,r_q)$; grid directory $\mathcal{K}$; target seed count $\tau$}
\KwOut{active seed set $I(q)\subseteq\Aset(c_q)$ with $|I(q)|\le\tau$}
$I\gets\emptyset$;\quad $\rho\gets 0$\;
\While{$|I|<\tau$}{
  $\mathcal{K}_\rho\gets$ grid cells at Chebyshev ring distance $\rho$ from the cell of $g_q$\;
  \lIf{$\mathcal{K}_\rho=\emptyset$ \textbf{or} $\min_{\kappa\in\mathcal{K}_\rho}\mathrm{mindist}(\kappa,g_q)>r_q$}{\textbf{break}}
  \ForEach{$\kappa\in\mathcal{K}_\rho$ \textbf{with} $\mathrm{mindist}(\kappa,g_q)\le r_q$}{
    \ForEach{$v\in\mathcal{K}[\kappa]$}{
      \lIf{$\dgeo(s_v,g_q)\le r_q$}{$I\gets I\cup\{v\}$}
      \lIf{$|I|=\tau$}{\textbf{break}}
    }
    \lIf{$|I|=\tau$}{\textbf{break}}
  }
  $\rho\gets\rho+1$\;
}
\Return $I$\;
\end{algorithm}

Algorithm~\ref{alg:seed} presents the spatial entry seeding procedure. The grid directory partitions latitude $[-90^\circ,90^\circ]$ and longitude $(-180^\circ,180^\circ]$ into cells of resolution $w_\phi = w_\lambda = 0.5^\circ$ ($360\times 720 = 259{,}200$ cells) and stores vertex identifiers in a contiguous array indexed by a prefix-sum offset table~\cite{Gaede_1998,Samet_1984}. Construction requires two linear counting passes ($<0.6$~s at $n=10^7$). At query time, Algorithm~\ref{alg:seed} expands Chebyshev rings $\rho=0,1,\ldots$ around the cell containing $g_q$, where $\mathrm{mindist}(\kappa,g_q)$ is the exact minimum great-circle distance from $g_q$ to cell $\kappa$, halting once $\tau=16$ active seeds are collected or $\min_{\kappa\in\mathcal{K}_\rho}\mathrm{mindist}(\kappa,g_q) > r_q$. The zero-auxiliary-storage seeding alternative evaluated in Exp-5 replaces the grid directory with $N_a=256$ anchor vertices selected via farthest-point sampling on the 3D unit sphere and descends the physical \method graph using $\dgeo(\cdot,g_q)$ until $\tau=16$ active vertices inside $\Omega_{c_q}$ are reached.

\subsection{Spherical Numerics, Hyperparameters, and Work Scaling}
\label{app:stream-numerics}
\label{app:geopool-detail}
\label{app:row-detail}
\label{app:repro}
\label{app:extra-exp}

\begin{table}[t]
  \centering
  \caption{\method construction hyperparameters, held fixed across all datasets and selectivities.}
  \label{tab:construction-params}
  \footnotesize
  \setlength{\tabcolsep}{3.5pt}
  \begin{tabular}{@{}lrl@{}}
    \toprule
    \textbf{Parameter} & \textbf{Value} & \textbf{Role in construction / search}\\
    \midrule
    $K_{\mathrm{vec}}$ & $192$ & NN-Descent vector candidate pool size\\
    $B_v$ & $128$ & global-vector source quota\\
    $B_g$ & $256$ & nearest-first geographic stream quota\\
    $\alpha$ & $1.2$ & intra-source and intra-fiber diversity parameter\\
    $C_{\mathrm{tie}}$ & $4096$ & per-fiber scan cap ($K_f=\min\{B_g,C_{\mathrm{tie}}\}$)\\
    $B$ & $384$ & fixed row width ($B=B_v+B_g$)\\
    $\tau$ & $16$ & active entry seeds fetched per query\\
    \bottomrule
  \end{tabular}
\end{table}

\begin{table}[t]
  \centering
  \caption{Impact of recovery budget $\rk$ on \textsc{Overture-1M} at $\bar{\sigma}=10^{-1}$ and fixed $L=16$ ($\rk=\infty$ scans all $B=384$ entries).}
  \label{tab:budget}
  \footnotesize
  \setlength{\tabcolsep}{4pt}
  \begin{tabular}{rrrrr}
    \toprule
    $\rk$ & Recall@10 & Vector calcs / query & Row scans / query & QPS\\
    \midrule
    $8$      & $0.3730$ & $113$   & $657$   & $16{,}864$\\
    $16$     & $0.5896$ & $207$   & $1{,}380$ & $9{,}667$\\
    $32$     & $0.7788$ & $405$   & $2{,}264$ & $5{,}441$\\
    $64$     & $0.8576$ & $845$   & $3{,}047$ & $3{,}243$\\
    $96$     & $0.8691$ & $1{,}330$ & $3{,}649$ & $2{,}257$\\
    $128$    & $0.8695$ & $1{,}813$ & $4{,}229$ & $1{,}734$\\
    $160$    & $0.8707$ & $2{,}294$ & $4{,}823$ & $1{,}405$\\
    $192$    & $0.8725$ & $2{,}775$ & $5{,}433$ & $1{,}179$\\
    $\infty$ & $0.8738$ & $4{,}157$ & $7{,}249$ & $797$\\
    \bottomrule
  \end{tabular}
\end{table}

\begin{table}[t]
  \centering
  \caption{Minimum mean vector distance calculations required to reach $\mathrm{Recall}@10 \ge 0.99$; values in parentheses show the ratio to $m_{\mathrm{tar}}=\lceil\bar{\sigma}n\rceil$ (``---'' indicates $0.99$ recall is not reached within the tested $L\le 2048$ cap).}
  \label{tab:work}
  \footnotesize
  \setlength{\tabcolsep}{2pt}
  \newcommand{\wk}[2]{#1\,\ensuremath{{\scriptscriptstyle(#2)}}}
  \begin{tabular}{lcccccc}
    \toprule
    \textbf{Dataset} & $10^{-5}$ & $10^{-4}$ & $10^{-3}$ & $10^{-2}$ & $10^{-1}$ & $1$\\
    \midrule
    \textsc{Overture-1M}    & \wk{10}{1.00} & \wk{95}{.95} & \wk{851}{.85}
                            & \wk{4746}{.47} & \wk{11194}{.11} & \wk{21329}{.02}\\
    \textsc{Yfcc-Clip}      & \wk{18}{1.76} & \wk{73}{.73} & \wk{484}{.48}
                            & \wk{1977}{.20} & \wk{2673}{.03} & \wk{3543}{.004}\\
    \textsc{Sift-Geo-1M}    & \wk{10}{1.00} & \wk{51}{.51} & \wk{321}{.32}
                            & \wk{2292}{.23} & \wk{5470}{.05} & \wk{3685}{.004}\\
    \textsc{SpaceV-Geo-10M} & \wk{90}{.90} & \wk{822}{.82} & \wk{7731}{.77}
                            & \wk{74975}{.75} & --- & \wk{38734}{.004}\\
    \bottomrule
  \end{tabular}
\end{table}

Each coordinate $s=(\lambda,\phi)$ is pre-converted into its 3D unit vector $\mathbf{u}_s = (\cos\phi\cos\lambda,\allowbreak \cos\phi\sin\lambda,\sin\phi)$. At query time, the disk predicate $\dgeo(s_v,g_q)\le r_q$ in Algorithm~\ref{alg:expand} is evaluated through a 32-bit latitude--longitude bounding box (inflated outward by $1$~ULP via \texttt{nextafter}) followed by one 3D dot product $\langle\mathbf{u}_{s_v},\mathbf{u}_{g_q}\rangle \ge \cos(r_q/R_{\oplus})$.

\noindent\textbf{Dataset Preparation and Baseline Tuning Protocol.}
On \textsc{Overture-1M}, text embeddings are computed with \texttt{MiniLM-L6-v2} over concatenated name, category, and address fields from Overture Maps release~2026-06-17.0 and unit-normalized in \texttt{float32}; on \textsc{Yfcc-Clip}, 512-dimensional OpenCLIP ViT-B/32 features~\cite{clip} are used directly without post-normalization. On \textsc{Sift-Geo-1M} and \textsc{SpaceV-Geo-10M}, synthetic coordinates are generated independently of the vectors with longitude uniform on $[-125^\circ,-66^\circ]$ and $\sin(\phi)$ uniform on $[\sin 24^\circ,\sin 50^\circ]$ (uniform surface area across the contiguous US), using NumPy PCG64 seeds \texttt{2026072203} (SIFT) and \texttt{2026072204} (SpaceV). Table~\ref{tab:construction-params} lists the fixed construction parameters of \method. For CubeGraph~\cite{https://doi.org/10.48550/arxiv.2604.06616}, we evaluate three profiles $(\text{layers},\allowbreak M,\allowbreak \mathit{efC},\allowbreak \text{cross-degree})$ with minimum leaf size $50$: $P_0=(6,16,200,2)$, $P_1=(6,32,400,4)$, and $P_2=(8,32,400,4)$, selecting $P_0$ on \textsc{Overture-1M} and $P_2$ on the other three datasets. For Geo-Partition, we sweep grid cell widths $\{1^\circ,0.5^\circ,0.25^\circ\}$ with cell-local HNSW parameters $M=32,\mathit{efC}=200$, selecting $0.25^\circ$ on \textsc{Overture-1M} and \textsc{Sift-Geo-1M} and $0.5^\circ$ on \textsc{Yfcc-Clip} and \textsc{SpaceV-Geo-10M}. Vector-First uses HNSW with $M=32,\mathit{efC}=200$, and ACORN~\cite{acorn} uses $M=32,\gamma=12,M_\beta=32$. At query time, \method sweeps $\rk\in\{8,16,32,\allowbreak 64,96,128,\allowbreak 160,192,\infty\}$ and $12$ beam sizes $L\in[16,1024]$ (extended to $4096$ at $\bar{\sigma}=10^{-2}$ and $2048$ at $\bar{\sigma}=10^{-1}$ on \textsc{SpaceV-Geo-10M}); CubeGraph, Geo-Partition, and Vector-First sweep $\mathit{EF}=64\cdot 2^i$ ($0\le i\le 4$) and over-fetch $10\cdot 4^j$ ($0\le j\le 5$) with intermediate points near the Pareto knee, while ACORN sweeps $\mathit{EF}$.

\noindent\textbf{Detailed Recovery-Budget and Query-Work Tables (Exp-4).}
Table~\ref{tab:budget} reports the full recovery-budget sweep on \textsc{Overture-1M} at $\bar{\sigma}=10^{-1}$ and fixed $L=16$, and Table~\ref{tab:work} reports the minimum mean vector distance calculations required to reach $\mathrm{Recall}@10\ge 0.99$ across all $24$ workload cells (complementing Exp-4 in Section~\ref{sec:experiments}). In Table~\ref{tab:budget}, raising $\rk$ from $96$ to $\infty$ changes Recall@10 by only $0.0047$ while tripling vector-distance evaluations and reducing QPS by $2.8\times$; across all $24$ cells, only the two unfiltered ($\bar{\sigma}=1$) cells on \textsc{SpaceV-Geo-10M} ($0.105$) and \textsc{Overture-1M} ($0.016$) exhibit a maximum-recall gap above $0.0032$ between $\rk=64$ and $\rk=\infty$. In Table~\ref{tab:work}, on all three $10^6$ datasets, as $\bar{\sigma}$ grows from $10^{-5}$ to $1$, the evaluated fraction of $\Aset(c_q)$ drops monotonically from $\sim 1.0$ (with \textsc{Yfcc-Clip}'s $1.76$ ratio at $10^{-5}$ caused by boundary coordinate ties making $m_q > m_{\mathrm{tar}}$) to $0.004$--$0.02$. On \textsc{SpaceV-Geo-10M}, consistent with the uniform-site transition analyzed in Exp-1(3), reaching $0.99$ recall evaluates $75\%$--$77\%$ of the active set at $\bar{\sigma}\in\{10^{-3},10^{-2}\}$ (up to $L=4096$), and at $\bar{\sigma}=10^{-1}$ ($m_{\mathrm{tar}}=10^6$) search reaches $0.9804$ recall at $\rk=64, L=2048$ ($99{,}366$ evaluations, or $0.10\,m_{\mathrm{tar}}$; $0.9823$ at $\rk=\infty$) and $0.9561$ recall in $23{,}766$ evaluations ($0.024\,m_{\mathrm{tar}}$), before global vector edges take over at $\bar{\sigma}=1$ ($38{,}734$ evaluations, or $0.004\,m_{\mathrm{tar}}$).

%% file: references.bib
@article{hnsw,
  author  = {Yu. A. Malkov and D. A. Yashunin},
  title   = {Efficient and Robust Approximate Nearest Neighbor Search Using Hierarchical Navigable Small World Graphs},
  journal = {IEEE Transactions on Pattern Analysis and Machine Intelligence},
  volume  = {42},
  number  = {4},
  pages   = {824--836},
  year    = {2020},
  doi     = {10.1109/TPAMI.2018.2889473}
}

@article{nsg,
  author  = {Cong Fu and Chao Xiang and Changxu Wang and Deng Cai},
  title   = {Fast Approximate Nearest Neighbor Search with the Navigating Spreading-out Graph},
  journal = {Proceedings of the VLDB Endowment},
  volume  = {12},
  number  = {5},
  pages   = {461--474},
  year    = {2019},
  doi     = {10.14778/3303753.3303754}
}

@article{diskann,
  numpages = {11},

  title={Diskann: Fast accurate billion-point nearest neighbor search on a single node},
  author={Jayaram Subramanya, Suhas and Devvrit, Fnu and Simhadri, Harsha Vardhan and Krishnawamy, Ravishankar and Kadekodi, Rohan},
  journal={Advances in neural information processing Systems},
  volume={32},
  year={2019}
}

@inproceedings{filteredDiskANN,
  author    = {Siddharth Gollapudi and Neel Karia and Varun Sivashankar and Ravishankar Krishnaswamy and Nikit Begwani and Swapnil Raz and Yiyong Lin and Yin Zhang and Neelam Mahapatro and Premkumar Srinivasan and Amit Singh and Harsha Vardhan Simhadri},
  title     = {Filtered-{DiskANN}: Graph Algorithms for Approximate Nearest Neighbor Search with Filters},
  booktitle = {Proceedings of the ACM Web Conference 2023},
  pages     = {3406--3416},
  publisher = {ACM},
  address   = {New York, NY, USA},
  year      = {2023},
  doi       = {10.1145/3543507.3583552}
}

@article{acorn,
  author    = {Liana Patel and Peter Kraft and Carlos Guestrin and Matei Zaharia},
  title     = {{ACORN}: Performant and Predicate-Agnostic Search Over Vector Embeddings and Structured Data},
  journal   = {Proceedings of the ACM on Management of Data},
  volume    = {2},
  number    = {3},
  articleno = {120},
  numpages  = {27},
  year      = {2024},
  doi       = {10.1145/3654923}
}

@article{serf,
  author    = {Chaoji Zuo and Miao Qiao and Wenchao Zhou and Feifei Li and Dong Deng},
  title     = {{SeRF}: Segment Graph for Range-Filtering Approximate Nearest Neighbor Search},
  journal   = {Proceedings of the ACM on Management of Data},
  volume    = {2},
  number    = {1},
  articleno = {69},
  numpages  = {26},
  year      = {2024},
  doi       = {10.1145/3639324}
}

@article{ung,
  author    = {Yuzheng Cai and Jiayang Shi and Yizhuo Chen and Weiguo Zheng},
  title     = {Navigating Labels and Vectors: A Unified Approach to Filtered Approximate Nearest Neighbor Search},
  journal   = {Proceedings of the ACM on Management of Data},
  volume    = {2},
  number    = {6},
  numpages  = {27},
  year      = {2024},
  doi       = {10.1145/3698822}
}

@inproceedings{window,
  author    = {Joshua Engels and Ben Landrum and Shangdi Yu and Laxman Dhulipala and Julian Shun},
  title     = {Approximate Nearest Neighbor Search with Window Filters},
  booktitle = {Proceedings of the 41st International Conference on Machine Learning},
  series    = {Proceedings of Machine Learning Research},
  volume    = {235},
  pages     = {12469--12490},
  publisher = {PMLR},
  address   = {Vienna, Austria},
  year      = {2024}
}

@misc{rnsg,
  author       = {Zhiqiu Zou and Ziqi Yin and Rong-Hua Li and Hongchao Qin and Qiangqiang Dai and Guoren Wang},
  title        = {{RNSG}: A Range-Aware Graph Index for Efficient Range-Filtered Approximate Nearest Neighbor Search},
  howpublished = {arXiv preprint arXiv:2603.12913},
  year         = {2026},
  doi          = {10.48550/arXiv.2603.12913}
}

@misc{jag,
  author       = {Haike Xu and Guy Blelloch and Laxman Dhulipala and Lars Gottesb{\"u}ren and Rajesh Jayaram and Jakub {\L}{\k{a}}cki},
  title        = {{JAG}: Joint Attribute Graphs for Filtered Nearest Neighbor Search},
  howpublished = {arXiv preprint arXiv:2602.10258},
  year         = {2026},
  doi          = {10.48550/arXiv.2602.10258}
}

@misc{urng,
  author       = {Siyuan Liang and Ziqi Yin and Qi Zhang and Ronghua Li and Guoren Wang and Kaiwen Xue and Daiyin Wang and Xubin Li},
  title        = {Efficient Graph Indexing for Interval-Aware Vector Search},
  howpublished = {arXiv preprint arXiv:2606.11789},
  year         = {2026},
  doi          = {10.48550/arXiv.2606.11789}
}

@misc{fiberann,
  author       = {Thuong Dang},
  title        = {Fiber-Navigable Search: A Geometric Approach to Filtered {ANN}},
  howpublished = {arXiv preprint arXiv:2604.00102},
  year         = {2026},
  doi          = {10.48550/arXiv.2604.00102}
}

@inproceedings{nndescent,
  publisher = {ACM},
  address = {New York, NY, USA},
  doi = {10.1145/1963405.1963487},

  author       = {Wei Dong and
                  Moses Charikar and
                  Kai Li},
  title        = {Efficient k-nearest neighbor graph construction for generic similarity
                  measures},
  booktitle    = {WWW},
  pages        = {577--586},
  year         = {2011}
}

@article{sift1m,
  author  = {Herv{\'e} J{\'e}gou and Matthijs Douze and Cordelia Schmid},
  title   = {Product Quantization for Nearest Neighbor Search},
  journal = {IEEE Transactions on Pattern Analysis and Machine Intelligence},
  volume  = {33},
  number  = {1},
  pages   = {117--128},
  year    = {2011},
  doi     = {10.1109/TPAMI.2010.57}
}

@InProceedings{clip,
  title = {Learning Transferable Visual Models From Natural Language Supervision},
  author = {Radford, Alec and Kim, Jong Wook and Hallacy, Chris and Ramesh, Aditya and Goh, Gabriel and Agarwal, Sandhini and Sastry, Girish and Askell, Amanda and Mishkin, Pamela and Clark, Jack and Krueger, Gretchen and Sutskever, Ilya},
  booktitle = {Proceedings of the 38th International Conference on Machine Learning},
  pages = {8748--8763},
  year = {2021},
  editor = {Meila, Marina and Zhang, Tong},
  volume = {139},
  series = {Proceedings of Machine Learning Research},
  month = {18--24 Jul},
  publisher = {PMLR},
  address = {Virtual},
  url = {https://proceedings.mlr.press/v139/radford21a.html}
}

@article{yfcc100m,
  author  = {Bart Thomee and David A. Shamma and Gerald Friedland and Benjamin Elizalde
             and Karl Ni and Douglas Poland and Damian Borth and Li-Jia Li},
  title   = {{YFCC100M}: The New Data in Multimedia Research},
  journal = {Communications of the ACM},
  volume  = {59},
  number  = {2},
  pages   = {64--73},
  year    = {2016},
  doi     = {10.1145/2812802}
}

@InProceedings{pmlr-v176-simhadri22a,
  title = {Results of the NeurIPS{’}21 Challenge on Billion-Scale Approximate Nearest Neighbor Search},
  author = {Simhadri, Harsha Vardhan and Williams, George and Aum\"uller, Martin and Douze, Matthijs and Babenko, Artem and Baranchuk, Dmitry and Chen, Qi and Hosseini, Lucas and Krishnaswamny, Ravishankar and Srinivasa, Gopal and Subramanya, Suhas Jayaram and Wang, Jingdong},
  booktitle = {Proceedings of the NeurIPS 2021 Competitions and Demonstrations Track},
  pages = {177--189},
  year = {2022},
  editor = {Kiela, Douwe and Ciccone, Marco and Caputo, Barbara},
  volume = {176},
  series = {Proceedings of Machine Learning Research},
  month = {06--14 Dec},
  publisher = {PMLR},
  address = {Virtual},
  url = {https://proceedings.mlr.press/v176/simhadri22a.html}
}

@article{Pan_2024, title={Survey of vector database management systems}, volume={33}, ISSN={0949-877X}, url={http://dx.doi.org/10.1007/s00778-024-00864-x}, DOI={10.1007/s00778-024-00864-x}, number={5}, journal={The VLDB Journal}, publisher={Springer Science and Business Media LLC}, author={Pan, James Jie and Wang, Jianguo and Li, Guoliang}, year={2024}, month={July}, pages={1591–1615} }

@misc{https://doi.org/10.48550/arxiv.2005.11401,
  doi = {10.48550/ARXIV.2005.11401},
  url = {https://arxiv.org/abs/2005.11401},
  author = {Lewis, Patrick and Perez, Ethan and Piktus, Aleksandra and Petroni, Fabio and Karpukhin, Vladimir and Goyal, Naman and Küttler, Heinrich and Lewis, Mike and Yih, Wen-tau and Rocktäschel, Tim and Riedel, Sebastian and Kiela, Douwe},
  title = {Retrieval-Augmented Generation for Knowledge-Intensive NLP Tasks},
  publisher = {arXiv},
  year = {2020},
  copyright = {arXiv.org perpetual, non-exclusive license}
}

@article{DBLP:journals/pvldb/WangXY021,
  author       = {Mengzhao Wang and
                  Xiaoliang Xu and
                  Qiang Yue and
                  Yuxiang Wang},
  title        = {A Comprehensive Survey and Experimental Comparison of Graph-Based
                  Approximate Nearest Neighbor Search},
  journal      = {Proc. {VLDB} Endow.},
  volume       = {14},
  number       = {11},
  pages        = {1964--1978},
  year         = {2021},
  url          = {http://www.vldb.org/pvldb/vol14/p1964-wang.pdf},
  doi          = {10.14778/3476249.3476255},
  bibsource    = {dblp computer science bibliography, https://dblp.org}
}

@article{DBLP:journals/pacmmod/AziziEP25,
  author       = {Ilias Azizi and
                  Karima Echihabi and
                  Themis Palpanas},
  title        = {Graph-Based Vector Search: An Experimental Evaluation of the State-of-the-Art},
  journal      = {Proc. {ACM} Manag. Data},
  volume       = {3},
  number       = {1},
  pages        = {43:1--43:31},
  year         = {2025},
  url          = {https://doi.org/10.1145/3709693},
  doi          = {10.1145/3709693},
  bibsource    = {dblp computer science bibliography, https://dblp.org}
}

@article{DBLP:journals/is/AumullerBF20,
  author       = {Martin Aum{\"{u}}ller and
                  Erik Bernhardsson and
                  Alexander John Faithfull},
  title        = {ANN-Benchmarks: {A} benchmarking tool for approximate nearest neighbor
                  algorithms},
  journal      = {Inf. Syst.},
  volume       = {87},
  pages        = {101374},
  year         = {2020},
  url          = {https://doi.org/10.1016/j.is.2019.02.006},
  doi          = {10.1016/J.IS.2019.02.006},
  bibsource    = {dblp computer science bibliography, https://dblp.org}
}

@article{Chronis_2025, title={Filtered Vector Search: State-of-the-Art and Research Opportunities}, volume={18}, ISSN={2150-8097}, url={http://dx.doi.org/10.14778/3750601.3750700}, DOI={10.14778/3750601.3750700}, number={12}, journal={Proceedings of the VLDB Endowment}, publisher={Association for Computing Machinery (ACM)}, author={Chronis, Yannis and Caminal, Helena and Papakonstantinou, Yannis and Özcan, Fatma and Ailamaki, Anastasia}, year={2025}, month=Aug, pages={5488–5492} }

@article{DBLP:journals/vldb/SongYCYXLC25,
  author       = {Yitong Song and
                  Bin Yao and
                  Zhida Chen and
                  Xin Yang and
                  Jiong Xie and
                  Feifei Li and
                  Mengshi Chen},
  title        = {Efficient top-k spatial-range-constrained approximate nearest neighbor
                  search on geo-tagged high-dimensional vectors},
  journal      = {{VLDB} J.},
  volume       = {34},
  number       = {1},
  pages        = {14},
  year         = {2025},
  url          = {https://doi.org/10.1007/s00778-024-00894-5},
  doi          = {10.1007/S00778-024-00894-5},
  bibsource    = {dblp computer science bibliography, https://dblp.org}
}

@misc{https://doi.org/10.48550/arxiv.2604.06616,
  doi = {10.48550/ARXIV.2604.06616},
  url = {https://arxiv.org/abs/2604.06616},
  author = {Yang, Mingyu and Li, Wentao and Wang, Wei},
  title = {CubeGraph: Efficient Retrieval-Augmented Generation for Spatial and Temporal Data},
  publisher = {arXiv},
  year = {2026},
  copyright = {Creative Commons Attribution Share Alike 4.0 International}
}

@article{Xu_2024, title={iRangeGraph: Improvising Range-dedicated Graphs for Range-filtering Nearest Neighbor Search}, volume={2}, ISSN={2836-6573}, url={http://dx.doi.org/10.1145/3698814}, DOI={10.1145/3698814}, number={6}, journal={Proceedings of the ACM on Management of Data}, publisher={Association for Computing Machinery (ACM)}, author={Xu, Yuexuan and Gao, Jianyang and Gou, Yutong and Long, Cheng and Jensen, Christian S.}, year={2024}, month=Dec, pages={1–26} }

@article{Liang_2024, title={UNIFY: Unified Index for Range Filtered Approximate Nearest Neighbors Search}, volume={18}, ISSN={2150-8097}, url={http://dx.doi.org/10.14778/3717755.3717770}, DOI={10.14778/3717755.3717770}, number={4}, journal={Proceedings of the VLDB Endowment}, publisher={Association for Computing Machinery (ACM)}, author={Liang, Anqi and Zhang, Pengcheng and Yao, Bin and Chen, Zhongpu and Song, Yitong and Cheng, Guangxu}, year={2024}, month=Dec, pages={1118–1130} }

@article{Peng_2025, title={Dynamic Range-Filtering Approximate Nearest Neighbor Search}, volume={18}, ISSN={2150-8097}, url={http://dx.doi.org/10.14778/3748191.3748193}, DOI={10.14778/3748191.3748193}, number={10}, journal={Proceedings of the VLDB Endowment}, publisher={Association for Computing Machinery (ACM)}, author={Peng, Zhencan and Qiao, Miao and Zhou, Wenchao and Li, Feifei and Deng, Dong}, year={2025}, month={June}, pages={3256–3268} }

@article{Jiang_2025, title={DIGRA: A Dynamic Graph Indexing for Approximate Nearest Neighbor Search with Range Filter}, volume={3}, ISSN={2836-6573}, url={http://dx.doi.org/10.1145/3725399}, DOI={10.1145/3725399}, number={3}, journal={Proceedings of the ACM on Management of Data}, publisher={Association for Computing Machinery (ACM)}, author={Jiang, Mengxu and Yang, Zhi and Zhang, Fangyuan and Hou, Guanhao and Shi, Jieming and Zhou, Wenchao and Li, Feifei and Wang, Sibo}, year={2025}, month={June}, pages={1–26} }

@inproceedings{Guttman_1984, series={SIGMOD ’84}, title={R-trees: a dynamic index structure for spatial searching}, url={http://dx.doi.org/10.1145/602259.602266}, DOI={10.1145/602259.602266}, booktitle={Proceedings of the 1984 ACM SIGMOD international conference on Management of data  - SIGMOD ’84}, publisher={ACM Press}, address={Boston, Massachusetts, USA}, author={Guttman, Antonin}, year={1984}, pages={47--57}, collection={SIGMOD ’84} }

@article{Bentley_1975, title={Multidimensional binary search trees used for associative searching}, volume={18}, ISSN={1557-7317}, url={http://dx.doi.org/10.1145/361002.361007}, DOI={10.1145/361002.361007}, number={9}, journal={Communications of the ACM}, publisher={Association for Computing Machinery (ACM)}, author={Bentley, Jon Louis}, year={1975}, month={Sept}, pages={509–517} }

@article{Samet_1984, title={The Quadtree and Related Hierarchical Data Structures}, volume={16}, ISSN={1557-7341}, url={http://dx.doi.org/10.1145/356924.356930}, DOI={10.1145/356924.356930}, number={2}, journal={ACM Computing Surveys}, publisher={Association for Computing Machinery (ACM)}, author={Samet, Hanan}, year={1984}, month={June}, pages={187–260} }

@article{Gaede_1998, title={Multidimensional access methods}, volume={30}, ISSN={1557-7341}, url={http://dx.doi.org/10.1145/280277.280279}, DOI={10.1145/280277.280279}, number={2}, journal={ACM Computing Surveys}, publisher={Association for Computing Machinery (ACM)}, author={Gaede, Volker and Günther, Oliver}, year={1998}, month={June}, pages={170–231} }

@article{Vapnik_1971,
  author  = {Vapnik, V. N. and Chervonenkis, A. Ya.},
  title   = {On the Uniform Convergence of Relative Frequencies of Events to Their Probabilities},
  journal = {Theory of Probability and Its Applications},
  volume  = {16},
  number  = {2},
  pages   = {264--280},
  year    = {1971},
  doi     = {10.1137/1116025}
}

@article{Sauer_1972,
  author  = {Sauer, N.},
  title   = {On the density of families of sets},
  journal = {Journal of Combinatorial Theory, Series A},
  volume  = {13},
  number  = {1},
  pages   = {145--147},
  year    = {1972},
  doi     = {10.1016/0097-3165(72)90019-2}
}

@book{Matousek_2002,
  address = {New York, NY, USA},

  author    = {Matou{\v{s}}ek, Ji{\v{r}}{\'\i}},
  title     = {Lectures on Discrete Geometry},
  series    = {Graduate Texts in Mathematics},
  volume    = {212},
  publisher = {Springer},
  year      = {2002},
  doi       = {10.1007/978-1-4613-0039-7}
}

@book{Penrose_2003,
  address = {Oxford, UK},

  author    = {Mathew D. Penrose},
  title     = {Random Geometric Graphs},
  series    = {Oxford Studies in Probability},
  volume    = {5},
  publisher = {Oxford University Press},
  year      = {2003},
  doi       = {10.1093/acprof:oso/9780198506263.001.0001}
}

@article{Penrose_1999,
  author  = {Penrose, Mathew D.},
  title   = {On {$k$}-Connectivity for a Geometric Random Graph},
  journal = {Random Structures \& Algorithms},
  volume  = {15},
  number  = {2},
  pages   = {145--164},
  year    = {1999},
  doi     = {10.1002/(SICI)1098-2418(199909)15:2<145::AID-RSA2>3.0.CO;2-G}
}

@article{Balister_2005,
  author  = {Balister, Paul and Bollob{\'a}s, B{\'e}la and Sarkar, Amites and Walters, Mark},
  title   = {Connectivity of Random {$k$}-Nearest-Neighbour Graphs},
  journal = {Advances in Applied Probability},
  volume  = {37},
  number  = {1},
  pages   = {1--24},
  year    = {2005},
  doi     = {10.1239/aap/1113402397}
}

@article{PenroseYang_2025,
  author  = {Penrose, Mathew D. and Yang, Xiaochuan},
  title   = {Fluctuations of the Connectivity Threshold and Largest Nearest-Neighbour Link},
  journal = {The Annals of Applied Probability},
  volume  = {35},
  number  = {6},
  pages   = {3906--3941},
  year    = {2025},
  doi     = {10.1214/25-AAP2210}
}

@book{DaleyVereJones_2003,
  author    = {Daryl J. Daley and David Vere-Jones},
  title     = {An Introduction to the Theory of Point Processes},
  edition   = {2},
  volume    = {1},
  series    = {Probability and Its Applications},
  publisher = {Springer},
  address   = {New York},
  year      = {2003},
  doi       = {10.1007/b97277}
}

@article{Antoniou_2010,
  author  = {Vyron Antoniou and Jeremy Morley and Muki Haklay},
  title   = {Web 2.0 Geotagged Photos: Assessing the Spatial Dimension of the Phenomenon},
  journal = {Geomatica},
  volume  = {64},
  number  = {1},
  pages   = {99--110},
  year    = {2010},
  url     = {https://www.sciencedirect.com/science/article/pii/S1195103624006360}
}

@article{sieve,
  author  = {Zhaoheng Li and Silu Huang and Wei Ding and Yongjoo Park and Jianjun Chen},
  title   = {{SIEVE}: Effective Filtered Vector Search with Collection of Indexes},
  journal = {Proceedings of the VLDB Endowment},
  volume  = {18},
  number  = {11},
  pages   = {4723--4736},
  year    = {2025},
  doi     = {10.14778/3749646.3749725}
}

@article{cgif,
  author  = {Jiarui Luo and Chaoji Zuo and Dong Deng},
  title   = {{CGIF}: Combining Proximity Graphs and Inverted Files for Efficient Filtered Vector Search over Arbitrary Predicates},
  journal = {Proceedings of the VLDB Endowment},
  volume  = {19},
  number  = {11},
  pages   = {3663--3675},
  year    = {2026},
  doi     = {10.14778/3836663.3836716}
}

@inproceedings{nhq,
  author    = {Mengzhao Wang and Lingwei Lv and Xiaoliang Xu and Yuxiang Wang and Qiang Yue and Jiongkang Ni},
  title     = {An Efficient and Robust Framework for Approximate Nearest Neighbor Search with Attribute Constraint},
  booktitle = {Advances in Neural Information Processing Systems},
  volume    = {36},
  pages     = {15738--15751},
  year      = {2023}
}

@article{deg,
  author  = {Ziqi Yin and Jianyang Gao and Pasquale Balsebre and Gao Cong and Cheng Long},
  title   = {{DEG}: Efficient Hybrid Vector Search Using the Dynamic Edge Navigation Graph},
  journal = {Proceedings of the ACM on Management of Data},
  volume  = {3},
  number  = {1},
  pages   = {1--28},
  year    = {2025}
}

@inproceedings{vbase,
  author    = {Qianxi Zhang and Shuotao Xu and Qi Chen and Guoxin Sui and Jiadong Xie and Zhizhen Cai and Yaoqi Chen and Yinxuan He and Yuqing Yang and Fan Yang and Mao Yang and Lidong Zhou},
  title     = {{VBASE}: Unifying Online Vector Similarity Search and Relational Queries via Relaxed Monotonicity},
  booktitle = {17th USENIX Symposium on Operating Systems Design and Implementation (OSDI 23)},
  pages     = {377--395},
  year      = {2023}
}

@article{chen2021survey,
  author  = {Zhida Chen and Lisi Chen and Gao Cong and Christian S. Jensen},
  title   = {Location- and Keyword-Based Querying of Geo-Textual Data: A Survey},
  journal = {The VLDB Journal},
  volume  = {30},
  number  = {4},
  pages   = {603--640},
  year    = {2021},
  doi     = {10.1007/s00778-021-00661-w}
}
